\documentclass[journal]{IEEEtran}
\usepackage{amssymb}
\usepackage{mathrsfs}
\usepackage{graphicx}
\usepackage{amsmath, amssymb, amsthm}
\usepackage{leftidx}
\usepackage{extarrows}
\usepackage{stfloats}
\usepackage{picinpar}
\usepackage{enumerate}
\usepackage{algorithm}
\usepackage{algorithmicx}
\usepackage{algpseudocode}
\usepackage{epsfig}
\usepackage{latexsym}
\usepackage{amsfonts}
\usepackage{enumerate}
\usepackage{graphics}
\usepackage{MnSymbol}
\usepackage{float}
\usepackage{pict2e}
\usepackage{tikz}
\newtheorem{theorem}{Theorem}

\newtheorem{lemma}{Lemma}
\newtheorem{remark}{Remark}
\newtheorem{definition}{Definition}

\newtheorem{proposition}{Proposition}

\usetikzlibrary{arrows}
\tikzset{
  treenode/.style = {align=center, inner sep=0pt, text centered,
    font=\sffamily},
  arn_n/.style = {treenode, circle, white, font=\sffamily\bfseries, draw=black,
    fill=black, text width=1.5em},
  arn_r/.style = {treenode, circle, red, draw=red,
    text width=1.5em, very thick},
  arn_x/.style = {treenode, rectangle, draw=black,
    minimum width=0.5em, minimum height=0.5em}
}

\title{Pinning Stabilizer Design for Large-Scale Probabilistic Boolean Networks}

\author{Lin Lin, Jinde Cao$^\dag$,~\IEEEmembership{Fellow,~IEEE}, Jianquan Lu,~\IEEEmembership{Senior Member,~IEEE}, and Jie Zhong
\thanks{Version October, 2020.}
\thanks{$^\dag$Corresponding author: Jinde Cao.}
\thanks{L. Lin and J. Lu are with the Jiangsu
Provincial Key Laboratory of Networked Collective Intelligence, the School of Mathematics, Southeast University, Nanjing 210096, China (email: linlin00wa@gmail.com; jqluma@seu.edu.cn).}
\thanks{J. Cao is with School of Mathematics, Southeast University, Nanjing, China, and with Yonsei Frontier Lab, Yonsei University, Seoul, Korea (email: jdcao@seu.edu.cn).}
\thanks{J. Zhong is with College of Mathematics and Computer Science, Zhejiang Normal University, Jinhua 321004, China (email: jiezhong0615math@zjnu.edu.cn).}}

\begin{document}
\maketitle
\thispagestyle{empty}
\pagestyle{empty}
\begin{abstract}
  This paper investigates the stabilization of probabilistic Boolean networks (PBNs) via a novel pinning control strategy based on network structure. In a PBN, the evolution equation of each gene switches among a collection of candidate Boolean functions with probability distributions that govern the activation frequency of each Boolean function.
  Owing to the stochasticity, the uniform state feedback controller, independent of switching signal, might be out of work, and in this case, the non-uniform state feedback controller is required. Subsequently, a criterion is derived to determine whether uniform controllers is applicable to achieve stabilization. It is worth pointing out that the pinning control designed in this paper is based on the network structure, which only requires local in-neighbors' information, rather than global information (state transition matrix). Moreover, this pinning control strategy reduces the computational complexity from $O(2^{2n})$ to $O(n2^\alpha)$, and thus it has the ability to handle some large-scale networks, especially the networks with sparse connections. Finally, the mammalian cell-cycle encountering a mutated phenotype is modelled by a PBN to demonstrate the obtained results.
\end{abstract}


\section{Introduction}\label{section-introduction}
A salient issue in biological regulatory networks is to properly understand the structure and temporal behaviour, which requires to integrate regulatory data into a formal dynamical model \cite{de2002modelling}. Although this problem has been recurrently solved by standard mathematical methods, such as differential or stochastic equations, it is still complicated owing to the diversity and sophistication of regulatory mechanisms, as well as the chronic lack of credible quantitative information \cite{faure2006dynamical}. Motivated by such defect, the intrinsically qualitative methods were developed to learn on Boolean algebra or generalisation thereof \cite{kauffman1969metabolic}. Boolean network (BN), as an effective approach for exploring the evolution patterns and structure, has increasingly attracted much interest. In a Boolean model, each node is valued as a binary logical variable, and its state is updated according to a certain specified logical rule composed of basic logical operators and its neighbors' states \cite{kauffman1969metabolic,azuma1,azuma2}.

Subsequently, probabilistic BNs (PBNs) were introduced in \cite{shmulevich2002probabilistic} to characterize the switch-like behavior of gene regulation networks. Such switch-like behavior is reflected when cells move from one state to another in a normal growth process and when cells respond to external signals. Particularly, switching in probability occurs in the discrete decision-making processes of the cell. To be more detailed, PBN is an effective tool to characterize the signal pathway of the mammalian cell-cycle with mutation phenotype \cite{faryabi2008regulatory}, which is further demonstrated in the simulation of this paper.

In recent years, Cheng and his cooperators \cite{chengdz2011springer} proposed an algebraic technique, called semi-tensor product (STP), for the analysis of Boolean (control) networks. The STP of matrices is defined as a novel matrix product for two arbitrary-dimensional matrices, thereby it breaks the traditional dimension-matching condition and is more flexible to utilize \cite{cheng2010analysis}. Based on the STP technique, several fundamental problems in control theory have been investigated for BNs and PBNs, including but not limited to, stability and stabilization \cite{zhusy2018tac,liht2017siam,liht2019siam}, controllability \cite{lu2018stabilization,Margaliot2018Auto}, observability \cite{margaliot2012aut1218,yyy2019observability,guoyq,guoyq2019auto}, synchronization \cite{chenhw2019}, optimal control \cite{wu2017policy,wuyhtac}, and other related problems \cite{lirui2018siam,yyyTAC,liuyang2017TAC}. In essence, the STP of matrices linearizes the algebraic function by enumerating the state space. It leads to a high computational complexity of many developed methods, and these methods are difficultly applied on large-dimensional BNs \cite{cheng2010analysis}.

Among the control-related issues, stabilization is a fundamental and essential problem in the therapeutic intervention and safety verification \cite{li2013state}. More precisely, a recent research discovers that gene activity emerges spontaneous and orderly collective behaviour \cite{huang1999gene}, and coincidentally, it can be properly demonstrated by the BN being stabilized to a certain state. In addition, in the long time evolution of genes in gene regulatory networks, steady states usually represent cell states, including cell death or unregulated growth. Hence, there is abundant justification in the assertion that one needs to design an efficient control strategy, under which the gene regulatory network is guided to a desirable state and remains at this state afterward.

In the last decades, the stabilization of BNs and PBNs has been studied mostly by means of traditional discrete control, including state feedback control \cite{lu2018stabilization,li2013state}, sampled-data control \cite{lu2018stabilization}, as well as event-triggered control \cite{zhusy2019}. These controllers are applied to either all the nodes or some randomly selected nodes of a BN; it may result in the greater control cost or redundant control inputs for some nodes.

Recently, one significative method, called pinning control strategy, has been introduced in \cite{lu2020tac} and has received considerable attentions. The main conception of pinning control is that only a fraction of nodes are determined to be imposed state feedback controllers, while the remaining nodes can be propagated through the coupling among nodes \cite{wu2013synchronization}. In the existing works, control design is almost from the point of state transition matrix \cite{liff2015pinning,liff2019set}. Unfortunately, such design approach will lead to a high-dimensional form of workable controllers associated with high computational complexity. For a BN with $n$ nodes, its state transition matrix is $2^{2n}$ dimension, and then the time complexity of designing pinning controllers is $O(2^{2n})$. In order to overcome this aporia, Zhong et al. \cite{zhong2020arxiv} developed a novel pinning control design strategy, which only utilized the local neighbors' information in the network structure, rather than the traditional global state transition space. Subsequently, Zhu et al. \cite{zhusy2020arxiv} further utilized such design method to address the observability problem for BNs. It is worth mentioning that this method successfully reduces the time complexity from $O(2^{2n})$ to $O(n2^{\alpha})$, where $\alpha$ is the largest in-degree of nodes in wring digraph. Hence, it provides a potential applications on some large-scale networks.

Inspired by the above discussions, we study the pinning control design for PBNs in consideration of the stochastic switching signal. The main contributions of this paper are concluded as follows:
\begin{enumerate}
  \item A more efficient pinning control is designed to stabilize PBNs with respect to the feature of each pinning node. More precisely, some pinning nodes can be imposed uniform state feedback controllers to reduce control cost; while other pinning nodes need to utilize different controllers for different possibilities. Particularly, a criterion is proposed to identify pinning nodes that can be applied on uniform controllers.
  \item The first group of pinning nodes (to achieve stabilization) is easily determined by finding a feasible feedback arc set (FAS) in the wiring digraph of the PBN. Besides, the second group of pinning nodes (to achieve stabilization at a certain state) can be selected by solving an optimization problem.
  \item The time complexity of our approach is dramatically reduced from $O(2^{2n})$ as in \cite{wanglq2019pinning,liff2019set} to $O(n2^{\max_{i\in[1,n]}|\mathcal{N}_i|})$, where $\mathcal{N}_i$ is the in-neighbor set of node $x_i$. Thereby, it is more implementable than conventional pinning control design methods when applied in large-scale networks, especially sparse networks.
\end{enumerate}

The reminder of this paper is organized as follows: Section \ref{section-preliminaries} introduces some basic definitions and model formulation. Sections \ref{section-further stability} and \ref{section-graphicaltheory} respectively design pinning control to achieve stabilization and stabilizing at a certain state. An biological example is established in Section \ref{section-conclusion}, followed by a brief conclusion in Section \ref{section-stabilization}.

\section{Preliminaries and model formulation}\label{section-preliminaries}

\subsection{Notations}
For the better expression, we list some basic notations, which are used throughout this study.
\begin{itemize}
  \item $\mathbb{N}_+$ is the set of positive integers;

  \item $[a,b]:=\{a,a+1,\cdots,b-1,b\}\subseteq \mathbb{N_+}$;

  \item $\mathscr{D}~:=\{0,1\}$ and
  $\mathscr{D}^n:=\underbrace{\mathscr{D}\times \mathscr{D} \times \cdots \times \mathscr{D}}_n$;

  \item $\text{Col}_j(A)$ (resp., $\text{Col}(A)$) is the $j$-th column (resp., column set) of matrix $A$;

  \item $\Delta_n:=\{\text{Col}_i(I_n)~|~i\in[1,n]\}$, where $I_n$ is the $n\times n$ identity matrix;

  \item $\mathscr{R}^{m\times n}$ (resp., $\mathscr{L}^{m\times n}$) is the set of $m\times n$ real matrices (resp., logical matrices). Moreover, matrix $A\in\mathscr{R}^{m\times n}$ is called a logical matrix, if $\text{Col}(A)\subseteq \Delta_{m}$;

  \item $\mathbf{r}:=(r_1, \cdots, r_p)^\top$ is a $p$-dimensional probability vector satisfying $r_i\geq 0,i\in[1,p]$ and $\Sigma_{i=1}^{p}r_i=1$;

  \item $\mathscr{P}^{m \times n}$ is the set of $m \times n$ probability matrices, whose columns are $m$-dimensional probability vectors;

  \item $\mathbb{B}_2$ is the set of $2$-dimensional probability vectors, where each element is larger than 0. Besides, $\mathbb{B}$ is the set of real numbers larger than 0 and less than 1;

  \item $[a_i]_{i\in I}:=a_{i_1},\cdots,a_{i_k}\in \{a_1,\cdots,a_n\}$ is a series of variables, where index set $I=\{i_1,\cdots,i_k\}\subseteq [1,n]$;

  \item $A^\top$ is the transposition of matrix $A$;

  \item $|\Omega|$ is the cardinal number of set $\Omega$.
\end{itemize}

\subsection{Problem Description}
PBN is a kind of logical network composed by certain number of nodes, in which each node is interacted by logical operators, including conjunction, disjunction, negation, and so on. Accordingly, a PBN with $n$ nodes considered in this paper can be described as:
\begin{equation}\label{equation-PBN}\left\{\begin{array}{l}
x_1(t+1)=f_1([x_j(t)]_{j\in \mathcal{N}_1}),\\
x_2(t+1)=f_2([x_j(t)]_{j\in \mathcal{N}_2}),\\
~~~~~~~~~~~~\vdots\\
x_n(t+1)=f_n([x_j(t)]_{j\in \mathcal{N}_n}),
\end{array}\right.
\end{equation}
where $x_i(t)\in \mathscr{D}, i\in [1,n]$ denotes the state of the $i$-th node at time instance $t$, set $\mathcal{N}_i\subseteq [1,n]$ contains the subscript indices of the in-neighbors of node $x_i$, $f_i:\mathscr{D}^{|\mathcal{N}_i|}\rightarrow \mathscr{D}, i\in[1,n]$ is the logical function for node $x_i$. Precisely, each logical function $f_i,i\in [1,n]$ for node $x_i$ has $l_i$ possibilities, and is chosen from a specific finite logical function set $\mathscr{F}_i=\{f_i^1,f_i^2,\cdots,f_i^{l_i}\}$. The probability of $f_i$ being $f_i^\kappa,\kappa\in[1,l_i]$ is assumed to be $p_i^\kappa$ with $\sum_{\kappa=1}^{l_i}p_i^\kappa=1$.
There are $\Pi_{i=1}^{n}l_{i}$ possibilities, and the $\lambda$-th model is denoted by $\Sigma_\lambda=\{f_1^{\lambda_1},f_2^{\lambda_2},\cdots,f_n^{\lambda_n}\}$, where $\lambda=\Sigma_{i=1}^{n-1}(\lambda_i-1)\Pi_{j=1}^{n-1}l_{j+1}+\lambda_n,\lambda_i \in [1,l_i]$. In this paper, the probability of each logical function is assumed to be mutually independent, then the probability of $\Sigma_\lambda$ being active is $P_\lambda=\Pi_{i=1}^{n}p_i^{\lambda_i}$. At each time step, $x(t)=(x_1(t),x_2(t),\cdots,x_n(t))\in \mathscr{D}^n$ is defined as the state of PBN (\ref{equation-PBN}).

\begin{definition}[Functional Variable]\label{functional-variable}
Consider logical function $f(x_1,\cdots,x_n):\mathscr{D}^{n}\rightarrow \mathscr{D}$, if $f(x_1,\cdots,x_i,\cdots,x_n)\neq f(x_1,\cdots,\neg x_i,\cdots,x_n)$, variable $x_i$ is called functional variable, otherwise it is called non-functional variable.
\end{definition}

Accordingly, denote $[x_j]_{j\in \mathcal{N}_{i,\kappa}}$ by functional variables for each possibility $f_i^\kappa\in \mathscr{F}_i,\kappa\in [1,l_i]$.
Thereupon, PBN (\ref{equation-PBN}) can be simplified by removing non-functional variables, and its network structure can be depicted by a wiring digraph defined as follows.

\begin{definition}[Wiring Digraph]\label{wiring-digraph}
The wiring digraph of a BN (or a PBN) with $n$ state components is denoted by $\mathbf{G}=(\mathbf{V},\mathbf{E})$. Thereinto, node set $\mathbf{V}$ is equivalent to $\{x_1,\cdots,x_n\}$, and edge set $\mathbf{E}$ consists of all edges $x_j\rightarrow x_i, i\in [1,n]$, where variable $x_j$ is a functional variable of logical function $f_i$ with positive probability.
\end{definition}

\begin{definition}[Global Stability]\label{global-stability}
The state $x(t)\in \mathscr{D}^n$ of PBN (\ref{equation-PBN}) is called a steady state if $P\{x(t+1)=x(t)\}=1$. Besides, $\{x(t),x(t+1),\cdots,x(t+s-1)\}\subseteq \mathscr{D}^n$ of PBN (\ref{equation-PBN}) is called an $s$-length cyclic attractor if $P\{x(t+s)=x(t)~\mathrm{and}~ x(t+\tau)\neq x(t), \tau\in [1,s-1]\}=1$. Then, PBN (\ref{equation-PBN}) is globally stable, if there exists a unique steady state as the attractor without cyclic attractor.
\end{definition}

\subsection{Model Transformation}
To convert the logical equation (\ref{equation-PBN}) into an algebraic representation, we introduce the STP tool. Moreover, its detailed properties and applications can be acquired in \cite{chengdz2011springer,cheng2010analysis}.

\begin{definition}[\cite{chengdz2011springer}]\label{stp}
Given two matrices $A\in \mathscr{R}^{m \times n}$ and $B\in \mathscr{R}^{p\times q}$, the STP of $A$ and $B$, termed as $A \ltimes B$, is defined as
\[A \ltimes B=(A\otimes I_{\frac{\alpha}{n}})(B\otimes I_{\frac{\alpha}{p}}).\]
Here, `$~\otimes$' is the Kronecker product, and $\alpha$ is the least common multiple of $n$ and $p$.
In addition, the symbol `$\ltimes$' is hereafter omitted if no confusion occurs.
\end{definition}

\begin{proposition}[\cite{chengdz2011springer}]\label{proposition-properties}
Some operation properties:
\begin{itemize}
  \item[1)] Let $A\in \mathscr{R}^{m \times n}$, $B\in \mathscr{R}^{p \times q}$ and $C\in \mathscr{R}^{r\times d}$, then $(A\ltimes B)\ltimes C = A\ltimes (B \ltimes C)$;

  \item[2)] Let $A\in \mathscr{R}^{p\times 1}$ and $B\in \mathscr{R}^{m\times n}$, then $A \ltimes B=(I_{p} \otimes B)\ltimes A$;

  \item[3)] Let $A\in \mathscr{R}^{p\times 1}$ and $B\in \mathscr{R}^{d\times 1}$, then $B \ltimes A= W_{[p,d]} \ltimes A \ltimes B$, where $W_{[p,d]}=[I_{d}\otimes \delta_{p}^{1}, \cdots, I_{d} \otimes \delta_{p}^{p}]\in \mathscr{L}^{pd\times pd}$ is called a swap matrix.
\end{itemize}
\end{proposition}

\begin{lemma}[\cite{chengdz2011springer}]\label{lemma-reducing matrix}
Let $\mathbf{x}=\ltimes_{i=1}^n \mathbf{x}_i=\mathbf{x}_1\ltimes \mathbf{x}_2\ltimes\ldots\ltimes \mathbf{x}_n$ with $\mathbf{x}_i\in\Delta_2, i\in[1,n]$, then it holds that $\mathbf{x}^2=\Phi_n \mathbf{x}$ with $\Phi_n=\Pi_{i=1}^n I_{2^{i-1}}\otimes [(I_2 \otimes W_{[2,2^{n-i}]})M_r]$, where $M_r$ denotes the power-reducing matrix satisfying $\mathbf{x}_i^2=M_r \mathbf{x}_i$.
\end{lemma}

\begin{definition}[\cite{RaoSolutions}]
Given two matrices $A\in \mathscr{R}^{m \times n}$ and $B\in \mathscr{R}^{w\times n}$, the Khatri-Rao product of $A$ and $B$, termed as $A\ast B$, is denoted by
\begin{equation*}
A\ast B=[\mathrm{Col}_1(A)\ltimes \mathrm{Col}_1(B),\cdots,\mathrm{Col}_n(A)\ltimes \mathrm{Col}_n(B)]\in \mathscr{R}^{mw \times n}.
\end{equation*}
\end{definition}

Denote a bijection $\sigma : \mathscr{D}\rightarrow \Delta_2$ as $\sigma(x)=\delta_2^{2-x}=\mathbf{x}$. Then, based on the STP of matrices, the algebraic representation for logical function $f:\mathscr{D}^n\rightarrow\mathscr{D}$ can be also derived.

\begin{lemma}[\cite{chengdz2011springer}]\label{lemma-structure matrix}
For a logical function $f(x_1, x_2,\ldots,x_n):\mathscr{D}^n\rightarrow \mathscr{D}$, its algebraic form $\mathbf{f}(\mathbf{x}_1, \mathbf{x}_2,\ldots,\mathbf{x}_n): (\Delta_2)^n\rightarrow \Delta_2$ satisfying
\begin{equation}
\mathbf{f}(\mathbf{x}_1, \mathbf{x}_2,\ldots,\mathbf{x}_n)=H_f \ltimes_{i=1}^n \mathbf{x}_i,
\end{equation}
where $H_f\in\mathscr{L}^{2\times 2^n}$ is unique and is called the structure matrix of logical function $f$.
\end{lemma}

According to Lemma \ref{lemma-structure matrix}, the structure matrix of logical function $f^\kappa_i,i\in [1,n]$ can be expressed as $F^\kappa_i\in \mathscr{L}^{2\times 2^{|\mathcal{N}_{i,\kappa}|}}$. Let $\widehat{F}_i^\kappa\in \mathscr{L}^{2\times 2^{|\mathcal{N}_{i}|}}$ be its extended structure matrix such that $\widehat{F}_i^\kappa \ltimes_{j\in \mathcal{N}_i}\mathbf{x}_j(t)=F_i^\kappa \ltimes_{j\in \mathcal{N}_{i,\kappa}}\mathbf{x}_j(t)$.  Then on the basis of the STP method, the algebraic representation of PBN (\ref{equation-PBN}) reads:
\begin{equation}\left\{\begin{array}{l}\label{equation-PBN-a1}
\mathbf{x}_1(t+1)=F_1\ltimes_{j\in \mathcal{N}_1}\mathbf{x}_j(t),\\
\mathbf{x}_2(t+1)=F_2\ltimes_{j\in \mathcal{N}_2}\mathbf{x}_j(t),\\
~~~~~~~~~~~~\vdots\\
\mathbf{x}_n(t+1)=F_n\ltimes_{j\in \mathcal{N}_n}\mathbf{x}_j(t),
\end{array}\right.
\end{equation}
where $F_i, i\in[1,n]$ is chosen from matrix set $\{\widehat{F}_i^\kappa\in \mathscr{L}^{2\times 2^{|\mathcal{N}_i|}}~|~\kappa\in [1,l_i]\}$ subject to the probability distribution $\{0\leq p_i^\kappa\leq1~|~\kappa\in [1,l_i]\}$. Furthermore, the mathematical expectation of $\mathbf{x}_i(t+1),i\in[1,n]$ is defined as
\begin{equation}\label{equation-PBN-algebraic}
\mathbb{E}\{\mathbf{x}_i(t+1)\}=\mathbb{E}\{F_i\ltimes_{j\in \mathcal{N}_i}\mathbf{x}_j(t)\}=\widehat{F}_i\mathbb{E}\{\ltimes_{j\in \mathcal{N}_i}\mathbf{x}_j(t)\},
\end{equation}
where $\widehat{F}_i=\sum_{\kappa=1}^{l_i}p_i^\kappa \widehat{F}_i^\kappa$, and $\mathbb{E}\{\cdot\}$ represents the mathematical expectation throughout this paper.

\begin{remark}
Reviewing the existing literatures \cite{liht2017siam,lu2018stabilization} and \cite{guoyq}, many fundamental results on stability, controllability and observability are generally based on state transition matrix $L\in \mathscr{P}^{2^n\times 2^n}$, which can be obtained by the following process:
\begin{itemize}
  \item Construct a set of matrices $\Upsilon_i \in \mathscr{P}^{2^{|\mathcal{N}_i|} \times {2^n}}, i\in [1,n]$, such that $\mathbb{E}\{\ltimes_{j\in \mathcal{N}_i}\mathbf{x}_j(t)\}=\Upsilon_i\mathbb{E}\{\ltimes_{j=1}^{n}\mathbf{x}_j(t)\}$.
  \item Obtain the corresponding augmented system:
\begin{equation}\begin{aligned}\label{equation-PBN-a2}
\mathbb{E}\{\mathbf{x}(t+1)\}=&\ltimes_{i=1}^n \widehat{F}_i\mathbb{E}\{\ltimes_{j\in \mathcal{N}_i}\mathbf{x}_j(t)\}\\
=&\ltimes_{i=1}^n \widehat{F}_i\Upsilon_i\mathbb{E}\{\ltimes_{j=1}^{n}\mathbf{x}_j(t)\}\\
\triangleq&L\mathbb{E}\{\mathbf{x}(t)\},
\end{aligned}
\end{equation}
where $L=(\widehat{F}_1\Upsilon_1)\ast \cdots \ast (\widehat{F}_n\Upsilon_n)\in \mathscr{P}^{2^n\times 2^n}$ is called the state transition matrix.
\end{itemize}
On the one hand, some control-related problems can be easily handled in virtue of the converted form (\ref{equation-PBN-a2}). On the other hand, the scale of matrix $L$ expands rapidly with the increase of network nodes. It reveals that designing controllers based on state transition matrix causes high computational complexity.
\end{remark}

\section{Pinning Control Design for Stabilizing PBNs}\label{section-further stability}
In this section, we investigate the stability criterion for PBN (\ref{equation-PBN}) based on network structure rather than state transition matrix.
The relationship between global stability of a prespecified PBN and the acyclic structure of its wiring digraph is derived based on the following Lemma \ref{lemma-robert}.
\begin{lemma} [\cite{discrete2012robert}]\label{lemma-robert}
A BN is globally stable, if its wiring digraph has no cycle.
\end{lemma}

It implies that in order to recover the global stability of BNs, we need to reconstruct an acyclic wiring digraph by deleting certain edges of original wiring diagraph.
This action can be achieved by imposing feasible control strategy on certain nodes as indicated in \cite{zhong2020arxiv} and \cite{zhusy2020arxiv}. Hereafter, we prove that Lemma \ref{lemma-robert} also holds for the case of PBNs.

\begin{proposition}\label{proposition-pbn}
If the wiring digraph of PBN (\ref{equation-PBN}) is acyclic, the wiring digraph of each possible model is also acyclic.
\end{proposition}

\begin{proof}
The probability of the $\lambda$-th model $\Sigma_\lambda=\{f_1^{\lambda_1},f_2^{\lambda_2},\cdots,f_n^{\lambda_n}\}$ being active is $P_\lambda=\Pi_{i=1}^{n}p_i^{\lambda_i}$. Then, for node $x_i,i\in[1,n]$, if $x_j$ is the functional variable of $f_i^{\lambda_i}$ with probability $p_i^{\lambda_i}$, there is an edge $x_j\rightarrow x_i$ in the wiring digraph of the $\lambda$-th model, and this edge inevitably exists in the wiring digraph of PBN (\ref{equation-PBN}). Thereby, if there is a cycle in one of $\Pi_{i=1}^{n}l_{i}$ possible models, this cycle must exist in PBN  (\ref{equation-PBN}), which conflicts with the hypothesis.
\end{proof}

\begin{theorem}
PBN (\ref{equation-PBN}) is globally stable, if its wiring digraph is acyclic and all possible models are stabilized at the same steady state.
\end{theorem}
\begin{proof}
If the wiring digraph of PBN (\ref{equation-PBN}) is acyclic, then according to Proposition \ref{proposition-pbn}, each possible model has no cycle. Moreover, in the light of Lemma \ref{lemma-robert}, each possible model is globally stable. Furthermore, in order to avoid multiple attractors, it should guarantee that all possible models are stable at the same steady state.
\end{proof}

In what follows, we design two kinds of state feedback controllers, that is, uniform one (independent of switching signal) and non-uniform one (related to each possible model), to stabilize PBN (\ref{equation-PBN}), before which the method of selecting pinning nodes is presented.
\subsection{Selecting Pinning Nodes}
Initially, we introduce the concept of feedback arc set.
\begin{definition}[Feedback Arc Set \cite{bang2008algorithems}]\label{feedback-arc-set}
Feedback arc set (FAS) is a subset of edges, deleting which the wiring diagraph possesses acyclic structure. In other words, an FAS contains all fixed points and at least one edge of each cycle.
\end{definition}

Fortunately, the fixed points and cycles existing in wiring diagraph $\mathbf{G}=(\mathbf{V},\mathbf{E})$ can be directly determined by the depth-first search algorithm. Accordingly, denote a possible FAS by $\{e_1,\cdots,e_\varpi\}\subseteq \mathbf{E}$. For edges $e_1,\cdots,e_\varpi \in \mathbf{E}$, their corresponding starting nodes and ending nodes are respectively denoted by $v_{-}(e_1),\cdots,v_{-}(e_\varpi)\in \mathbf{V}$ and $v_{+}(e_1),\cdots,v_{+}(e_\varpi)\in \mathbf{V}$. With respect to the subgraph
$\mathbf{G}_i=(\mathbf{V}_i,\mathbf{E}_i)\subseteq\mathbf{G},i\in [1,n]$ induced by $x_i$, it reveals that $\bigcup_{e\in\mathbf{E}_i}v_{-}(e)=\bigcup_{\kappa=1}^{l_i}\bigcup_{j\in\mathcal{N}_{i,\kappa}}x_j(t)$ and $\bigcup_{e\in\mathbf{E}_i}v_{+}(e)=x_i(t)$. Moreover, subgraph $\mathbf{G}_i$ has the property that each starting node is connected by a certain edge, but each ending node may be connected by several different edges.

To proceed, we suppose that $\bigcup_{j=1}^{\varpi}v_{+}(e_j)=\{x_{\gamma_1},x_{\gamma_2},\cdots,x_{\gamma_\tau}\}\subseteq \{x_1,\cdots,x_n\}$, then it has
\begin{equation*}
\bigcup_{j=\omega_1^1}^{\omega_1^{\alpha_1}}v_{+}(e_j)=x_{\gamma_1},
\bigcup_{j=\omega_2^1}^{\omega_2^{\alpha_2}}v_{+}(e_j)=x_{\gamma_2},\cdots,
\bigcup_{j=\omega_\tau^1}^{\omega_\tau^{\alpha_\tau}}v_{+}(e_j)=x_{\gamma_\tau},
\end{equation*}
where $\{\omega_1^1,\cdots,\omega_1^{\alpha_1},\cdots,\omega_\tau^1,\cdots,\omega_\tau^{\alpha_\tau}\}=[1,\varpi]$ with $\omega_1^1<\cdots<\omega_1^{\alpha_1}<\cdots<\omega_\tau^1<\cdots<\omega_\tau^{\alpha_\tau}$. Correspondingly, the starting nodes are classified as
\[\bigcup_{j=\omega_1^1}^{\omega_1^{\alpha_1}}v_{-}(e_j)=\{x_{\pi_1^1},\cdots,x_{\pi_1^{\alpha_1}}\},
\bigcup_{j=\omega_2^1}^{\omega_2^{\alpha_2}}v_{-}(e_j)=\{x_{\pi_2^1},\cdots,x_{\pi_2^{\alpha_2}}\}, \]
\[\cdots,\bigcup_{j=\omega_\tau^1}^{\omega_\tau^{\alpha_\tau}}v_{-}(e_j)=\{x_{\pi_\tau^1},\cdots,x_{\pi_\tau^{\alpha_\tau}}\},\]
where $\pi_1^1,\cdots,\pi_1^{\alpha_1},\cdots,\pi_\tau^1,\cdots,\pi_\tau^{\alpha_\tau}$ with $\sum_{j=1}^\tau \alpha_j=\varpi$.

Afterwards, we reconstruct the acyclic wiring digraph based on the obtained FAS $\{e_1,\cdots,e_\varpi\}$, edges in which connect starting nodes $x_{\pi^1_1},\cdots,x_{\pi_1^{\alpha_1}}$ $,\cdots,x_{\pi^1_\tau},\cdots,x_{\pi_\tau^{\alpha_\tau}}\in \mathbf{V}$ and ending nodes $x_{\gamma_1},\cdots,x_{\gamma_\tau}\in \mathbf{V}$. It should be pointed out that all ending nodes are chosen to be pinning nodes, then the first group of pinning nodes (to achieve stabilization) are $x_{\gamma_1},\cdots,x_{\gamma_\tau}$, whose subscript indices are collected into set $\Lambda$.
For each pinning node $x_{\gamma_\sigma},\sigma\in[1,\tau]$, the subscript indices of its starting nodes are collected into set $\mathcal{N}^\star_{\gamma_\sigma}=\{\pi_\sigma^1,\cdots,\pi_\sigma^{\alpha_\sigma}\}$, and denote $\mathcal{N}^\circ_{\gamma_\sigma}=\mathcal{N}_{\gamma_\sigma}\backslash\mathcal{N}^\star_{\gamma_\sigma}$. Then, there exists a pair of probability matrices $\mathbf{T}_{\gamma_\sigma} \in \mathscr{P}^{2\times 2^{|\mathcal{N}_{\gamma_\sigma}|}}$ and $G_{\gamma_\sigma} \in \mathscr{P}^{2\times 2^{|\mathcal{N}^\circ_{\gamma_\sigma}|}}$ such that
\begin{equation}\label{stable-struture-matrix-0}
\mathbf{T}_{\gamma_\sigma}=G_{\gamma_\sigma}(I_{2^{|\mathcal{N}^\circ_{\gamma_\sigma}|}}\otimes \mathbf{1}^\top_{2^{|\mathcal{N}^\star_{\gamma_\sigma}|}}).
\end{equation}
\begin{remark}
Notice that there are many available matrix pairs ($\mathbf{T}_{\gamma_\sigma}$ and $G_{\gamma_\sigma}$) satisfying equations (\ref{stable-struture-matrix-0}), and they correspond to different control inputs. However, we can choose arbitrary one set of effective solutions for further analysis.
\end{remark}

Hereby, we transform the evolution equation of pinning node $x_{\gamma_\sigma},\sigma\in[1,\tau]$ from (\ref{equation-PBN-algebraic}) to the following form:
\begin{equation}\begin{aligned}\label{Ex(t+1)}
\mathbb{E}\{\mathbf{x}_{\gamma_\sigma}(t+1)\}&=\widehat{F}_{\gamma_\sigma}\mathbb{E}\{\ltimes_{j\in \mathcal{N}_{\gamma_\sigma}}\mathbf{x}_j(t)\}\\
&=\widehat{F}_{\gamma_\sigma}\mathbb{E}\{\mathbf{W}_{\gamma_\sigma}\ltimes_{j\in \mathcal{N}^\circ_{\gamma_\sigma}}\mathbf{x}_j(t)\ltimes_{j\in \mathcal{N}^\star_{\gamma_\sigma}}\mathbf{x}_j(t)\}\\
&\triangleq \mathbf{L}_{\gamma_\sigma}\mathbb{E}\{\ltimes_{j\in \mathcal{N}^\circ_{\gamma_\sigma}}\mathbf{x}_j(t)\ltimes_{j\in \mathcal{N}^\star_{\gamma_\sigma}}\mathbf{x}_j(t)\},
\end{aligned}
\end{equation}
where $\mathbf{W}_{\gamma_\sigma}=[\ltimes_{\sigma=1}^{\alpha_j}W_{[2,2^{|\pi_j^\sigma|\mathcal{N}_{\gamma_\sigma}}]}W_{[2^{|\mathcal{N}^\star_{\gamma_\sigma}|},2^{|\mathcal{N}^\circ_{\gamma_\sigma}|}]}]\otimes I_{2^{|\mathcal{N}_{\gamma_\sigma}|-|\mathcal{N}^\star_{\gamma_\sigma}|}}\in \mathscr{L}^{2^{|\mathcal{N}_{\gamma_\sigma}|}\times2^{|\mathcal{N}_{\gamma_\sigma}|}}$, and $\mathbf{L}_{\gamma_\sigma}\triangleq \widehat{F}_{\gamma_\sigma}\mathbf{W}_{\gamma_\sigma}\in \mathscr{P}^{2\times2^{|\mathcal{N}_{\gamma_\sigma}|}}$.

\subsection{Design Uniform State Feedback Controllers}
To achieve global stabilization for PBN (\ref{Ex(t+1)}), the state feedback controller imposed on pinning node $x_{\gamma_\sigma},\sigma\in [1,\tau]$ is designed as follows:
\begin{equation}\label{state-feedback-pinning-controllers}
u_{\gamma_\sigma}(t)=\varphi_{\gamma_\sigma}([x_j(t)]_{j\in \mathcal{N}_{\gamma_\sigma}}),
\end{equation}
where $u_{\gamma_\sigma}(t)\in \mathscr{D}$ denotes control input, and $\varphi_{\gamma_\sigma}:\mathscr{D}^{{|\mathcal{N}_{\gamma_\sigma}|}}\rightarrow\mathscr{D}$ is the logical function determined by node $x_{\gamma_\sigma}$'s in-neighbors. Subsequently, the pinning controlled PBN is described as:
\begin{equation}x_i(t+1)=\left\{\begin{array}{ll}\label{equation-BCN}
f_i([x_j(t)]_{j\in \mathcal{N}_i}),&i\in [1,n]\backslash \Lambda,\\
u_i(t) \odot_i f_i([x_j(t)]_{j\in \mathcal{N}_i}),&i\in\Lambda,
\end{array}\right.
\end{equation}
where $\odot_i:\mathscr{D}^2\rightarrow \mathscr{D}, i\in \Lambda$ is logical function connecting state feedback controller $u_i(t)$ and original dynamic equation $f_i([x_j(t)]_{j\in \mathcal{N}_i})$.
Furthermore, by resorting to the STP tool, the structure matrices of $\varphi_{\gamma_\sigma}$ and $\odot_{\gamma_\sigma}$ are respectively denoted by $\Psi_{\gamma_\sigma}\in \mathscr{L}^{2\times2^{|\mathcal{N}_{\gamma_\sigma}|}}$ and $M_{\odot_{\gamma_\sigma}}\in \mathscr{L}^{2\times 4}$, where $\gamma_\sigma\in\Lambda$. Then one derives that
\begin{equation}\begin{aligned}
\mathbf{u}_{\gamma_\sigma}(t)=&\Psi_{\gamma_\sigma}\ltimes_{j\in \mathcal{N}_{\gamma_\sigma}}\mathbf{x}_j(t)\\
=&\Psi_{\gamma_\sigma}\mathbf{W}_{\gamma_\sigma}\ltimes_{j\in \mathcal{N}^\circ_{\gamma_\sigma}}\mathbf{x}_j(t)\ltimes_{j\in \mathcal{N}^\star_{\gamma_\sigma}}\mathbf{x}_j(t)\\
\triangleq& \widehat{\Psi}_{\gamma_\sigma}\ltimes_{j\in \mathcal{N}^\circ_{\gamma_\sigma}}\mathbf{x}_j(t)\ltimes_{j\in \mathcal{N}^\star_{\gamma_\sigma}}\mathbf{x}_j(t),
\end{aligned}
\end{equation}
where $\widehat{\Psi}_{\gamma_\sigma}\triangleq \Psi_{\gamma_\sigma}\mathbf{W}_{\gamma_\sigma}\in \mathscr{L}^{2\times2^{|\mathcal{N}_{\gamma_\sigma}|}}$.
Consequently, the state updating of each pinning controlled node $x_{\gamma_\sigma},\sigma\in [1,\tau]$ is described in the algebraic form:
\begin{equation}\begin{aligned}
\mathbb{E}&\{\mathbf{x}_{\gamma_\sigma}(t+1)\}=M_{\odot_{\gamma_\sigma}}\mathbf{u}_{\gamma_\sigma}(t)\widehat{F}_i\mathbb{E}\{\ltimes_{j\in \mathcal{N}_i}\mathbf{x}_j(t)\}\\
=&M_{\odot_{\gamma_\sigma}}\widehat{\Psi}_{\gamma_\sigma}(I_{2^{|\mathcal{N}_{\gamma_\sigma}|}}\otimes \mathbf{L}_{\gamma_\sigma}) \Phi_{|\mathcal{N}_{\gamma_\sigma}|}\mathbb{E}\{\ltimes_{j\in\mathcal{N}^\circ_{\gamma_\sigma}}\mathbf{x}_j(t)\ltimes_{j\in \mathcal{N}^\star_{\gamma_\sigma}}\mathbf{x}_j(t)\}.
\end{aligned}
\end{equation}
If pinning controlled PBN (\ref{equation-BCN}) achieves global stabilization, it satisfies that
\begin{equation}\label{stable-structure-matrix}\begin{aligned}
M_{\odot_{\gamma_\sigma}}\widehat{\Psi}_{\gamma_\sigma}(I_{2^{|\mathcal{N}_{\gamma_\sigma}|}}\otimes \mathbf{L}_{\gamma_\sigma}) \Phi_{|\mathcal{N}_{\gamma_\sigma}|}&=G_{\gamma_\sigma}(I_{2^{|\mathcal{N}^\circ_{\gamma_\sigma}|}}\otimes \mathbf{1}^\top_{2^{|\mathcal{N}^\star_{\gamma_\sigma}|}})\\
&:=\mathbf{T}_{\gamma_\sigma}.
\end{aligned}
\end{equation}
Therefore, the effective state feedback controller (or structure matrices $M_{\odot_{\gamma_\sigma}}\in \mathscr{L}^{2\times 4}$ and $\widehat{\Psi}_{\gamma_\sigma}\in \mathscr{L}^{2\times 2^{|\mathcal{N}_{\gamma_\sigma}|}}$) for pinning node $x_{\gamma_\sigma}$ can be obtained by solving the above equations. Particularly, if equation (\ref{stable-structure-matrix}) is solvable,
logical functions $\odot_{\gamma_\sigma}$ and $\varphi_{\gamma_\sigma}$ can be uniquely derived according to Lemma \ref{lemma-structure matrix}.
Subsequently, we study the solvability criterion of equation (\ref{stable-structure-matrix}), before which several subsets are constructed as follows:
\begin{equation}\begin{aligned}\label{theorem-condition}
\Omega_1&=\{\mathrm{Col}_i(\mathbf{T}_{\gamma_\sigma})~|~\mathrm{Col}_i(\mathbf{T}_{\gamma_\sigma})=\mathrm{Col}_i(\mathbf{L}_{\gamma_\sigma})\},\\
\Omega_2&=\{\mathrm{Col}_i(\mathbf{T}_{\gamma_\sigma})~|~\mathrm{Col}_i(\mathbf{T}_{\gamma_\sigma})=\mathrm{Col}_i(\mathbf{L}_{\gamma_\sigma})=(1,0)^\top;~\mathrm{or}~\\
&~~~~\mathrm{Col}_i(\mathbf{T}_{\gamma_\sigma})=(1,0)^\top,\mathrm{Col}_i(\mathbf{L}_{\gamma_\sigma})=(0,1)^\top;~\mathrm{or}~\\
&~~~~\mathrm{Col}_i(\mathbf{T}_{\gamma_\sigma})=(1,0)^\top,\mathrm{Col}_i(\mathbf{L}_{\gamma_\sigma})\in \mathbb{B}_2\},\\
\Omega_3&=\{\mathrm{Col}_i(\mathbf{T}_{\gamma_\sigma})~|~\mathrm{Col}_i(\mathbf{T}_{\gamma_\sigma})=(1,1)^\top-\mathrm{Col}_i(\mathbf{L}_{\gamma_\sigma})\},\\
\Omega_4&=\{\mathrm{Col}_i(\mathbf{T}_{\gamma_\sigma})~|~\mathrm{Col}_i(\mathbf{T}_{\gamma_\sigma})=\mathrm{Col}_i(\mathbf{L}_{\gamma_\sigma})=(0,1)^\top;~\mathrm{or}~\\
&~~~~~\mathrm{Col}_i(\mathbf{T}_{\gamma_\sigma})=(0,1)^\top,\mathrm{Col}_i(\mathbf{L}_{\gamma_\sigma})=(1,0)^\top;~\mathrm{or}~\\
&~~~~~\mathrm{Col}_i(\mathbf{T}_{\gamma_\sigma})=(0,1)^\top,\mathrm{Col}_i(\mathbf{L}_{\gamma_\sigma})\in \mathbb{B}_2\}.
\end{aligned}
\end{equation}

\begin{theorem} \label{theorem-2-condition}
Equations (\ref{stable-structure-matrix}) is solvable for pinning node $x_{\gamma_\sigma},\sigma\in [1,\tau]$ if and only if $\mathrm{Col}(\mathbf{T}_{\gamma_\sigma})$ can be completely covered by at most two of $\Omega_1,\Omega_2,\Omega_3$ and $\Omega_4$.
\end{theorem}

\begin{proof}
(Sufficiency.) With respect to equation (\ref{stable-structure-matrix}), we presume its parameters as
\begin{equation*}M_{\odot_{\gamma_\sigma}}=\left(\begin{array}{cccc}
\alpha_1 & \alpha_2 & \alpha_3 & \alpha_4 \\
1-\alpha_1 & 1-\alpha_2 & 1-\alpha_3 & 1-\alpha_4
\end{array}\right),
\end{equation*}

\begin{equation*}\widehat{\Psi}_{\gamma_\sigma}=\left(\begin{array}{cccc}
\beta_1 & \beta_2 & \cdots & \beta_{2^{|\mathcal{N}_{\gamma_\sigma}|}} \\
1-\beta_1 & 1-\beta_2 & \cdots & 1-\beta_{2^{|\mathcal{N}_{\gamma_\sigma}|}}
\end{array}\right),
\end{equation*}

\begin{equation*}\mathbf{L}_{\gamma_\sigma}=\left(\begin{array}{cccc}
\xi_1 & \xi_2 & \cdots & \xi_{2^{|\mathcal{N}_{\gamma_\sigma}|}} \\
1-\xi_1 & 1-\xi_2 & \cdots & 1-\xi_{2^{|\mathcal{N}_{\gamma_\sigma}|}}
\end{array}\right),
\end{equation*}
\begin{equation*}\mathrm{and}~\mathbf{T}_{\gamma_\sigma}=\left(\begin{array}{cccc}
\eta_1 & \eta_2 & \cdots & \eta_{2^{|\mathcal{N}_{\gamma_\sigma}|}} \\
1-\eta_1 & 1-\eta_2 & \cdots & 1-\eta_{2^{|\mathcal{N}_{\gamma_\sigma}|}}
\end{array}\right),
\end{equation*}
where $\alpha_1,\alpha_2,\alpha_3,\alpha_4,\beta_1,\cdots,\beta_{2^{|\mathcal{N}_{\gamma_\sigma}|}}\in \mathscr{D}$ and $\xi_1,\cdots,\xi_{2^{|\mathcal{N}_{\gamma_\sigma}|}},$ $\eta_1,\cdots,\eta_{2^{|\mathcal{N}_{\gamma_\sigma}|}}\in \mathbb{B}\bigcup\{0,1\}$. After that, equations (\ref{stable-structure-matrix}) is converted into the following equations:
\begin{equation}\left\{\small{\begin{aligned}\label{equations-BCN-proof}
[\alpha_1\beta_1+\alpha_3(1-\beta_1)]\xi_1+&[\alpha_2\beta_1+\alpha_4(1-\beta_1)](1-\xi_1)=\eta_1,\\
[\alpha_1\beta_2+\alpha_3(1-\beta_2)]\xi_2+&[\alpha_2\beta_2+\alpha_4(1-\beta_2)](1-\xi_2)=\eta_2,\\
&\vdots \\
[\alpha_1\beta_i+\alpha_3(1-\beta_i)]\xi_i+&[\alpha_2\beta_i+\alpha_4(1-\beta_i)](1-\xi_i)=\eta_i,\\
&\vdots \\
[\alpha_1\beta_{2^{|\mathcal{N}_{\gamma_\sigma}|}}+\alpha_3(1-&\beta_{2^{|\mathcal{N}_{\gamma_\sigma}|}})]
\xi_{2^{|\mathcal{N}_{\gamma_\sigma}|}}+\\
[\alpha_2\beta_{2^{|\mathcal{N}_{\gamma_\sigma}|}}+\alpha_4(1-&\beta_{2^{|\mathcal{N}_{\gamma_\sigma}|}})](1-\xi_{2^{|\mathcal{N}_{\gamma_\sigma}|}})=\eta_{2^{|\mathcal{N}_{\gamma_\sigma}|}}.
\end{aligned}}\right.
\end{equation}
In what follows, we discuss the solution of the $i$-th ($i\in[1,2^{|\mathcal{N}_{\gamma_\sigma}|}]$) equation from the following three cases:
$$(a)~\eta_i=\xi_i;~~(b)~\eta_i=1-\xi_i;~~(c)~\eta_i\in \{0,1\},\xi\in \mathbb{B}.$$
First, we prove that cases $(a)\sim(c)$ have covered all situations. Since $\beta_i\in \mathscr{D}$, it derives $\alpha_3\xi_i+\alpha_4(1-\xi_i)=\eta_i$ when $\beta_i=0$; and $\alpha_1\xi_i+\alpha_2(1-\xi_i)=\eta_i$ when $\beta_i=1$. Owing to $\alpha_1,\alpha_2,\alpha_3,\alpha_4\in \mathscr{D}$, we can classify it into four situations:
(i) $\alpha_3=\alpha_4=0$; or $\alpha_1=\alpha_2=0$,
(ii) $\alpha_3=0,\alpha_4=1$; or $\alpha_1=0,\alpha_2=1$,
(iii) $\alpha_3=1,\alpha_4=0$; or $\alpha_1=1,\alpha_2=0$, and
(iv) $\alpha_3=\alpha_4=1$; or $\alpha_1=\alpha_2=1$, which respectively derive
(i) $\eta_i=0$,
(ii) $\eta_i=1-\xi_i$,
(iii) $\eta_i=\xi_i$,
and (iv) $\eta_i=1$. Moreover, case $(c)$ means $\eta_i\neq\xi_i$ and $\eta_i\neq 1-\xi_i$, then it holds $0<\xi_i<1$ since $\eta_i$ only equals to 0 or 1.

For case $(a)$, if $\beta_i=0$, one has $\alpha_3\xi_i+\alpha_4(1-\xi_i)=\xi_i$, that is, $(1-\alpha_3+\alpha_4)\xi_i=\alpha_4$. There are two solutions:
\begin{equation*}
(a1)~1-\alpha_3+\alpha_4=\alpha_4=0;~~~
(a2)~\xi_i=\frac{\alpha_4}{1-\alpha_3+\alpha_4},
\end{equation*}
which derives
\begin{equation*}\begin{array}{l}
(a1')~\xi_i=\eta_i, \alpha_3=1,\beta_i=\alpha_4=0; \\
(a2')~\xi_i=\eta_i=0,\beta_i=\alpha_4=0,\alpha_3=0;\\
(a3')~\xi_i=\eta_i=1,\beta_i=0,\alpha_3=\alpha_4=1;\\
(a4')~\xi_i=\eta_i=\frac{1}{2},\beta_i=\alpha_3=0,\alpha_4=1.
\end{array}
\end{equation*}
If $\beta_i=1$, it has $\alpha_1\xi_i+\alpha_2(1-\xi_i)=\xi_i$ and derives
\begin{equation*}\begin{array}{l}
(a1'')~\xi_i=\eta_i,\beta_i=\alpha_1=1,\alpha_2=0; \\ (a2'')~\xi_i=\eta_i=0,\beta_i=1,\alpha_2=0,\alpha_1=0;\\
(a3'')~\xi_i=\eta_i=1,\beta_i=\alpha_1=\alpha_2=1;\\
(a4'')~\xi_i=\eta_i=\frac{1}{2},\alpha_1=0,\beta_i=\alpha_2=1.
\end{array}
\end{equation*}

For case $(b)$, if $\beta_i=0$, one acquires $\alpha_3\xi_i+\alpha_4(1-\xi_i)=1-\xi_i$. It also has two solutions:
\begin{equation*}
(b1)~\alpha_3-\alpha_4+1=1-\alpha_4=0;~~~
(b2)~\xi_i=\frac{1-\alpha_4}{\alpha_3-\alpha_4+1},
\end{equation*}
which obtains
\begin{equation*}\begin{array}{ll}
(b1')~\eta_i=1-\xi_i,\beta_i=\alpha_3=0,\alpha_4=1;\\ (b2')~\eta_i=1-\xi_i=\frac{1}{2},\alpha_3=1,\beta_i=\alpha_4=0;\\
(b3')~\eta_i=1-\xi_i=1,\beta_i=0,\alpha_3=\alpha_4=1;\\
(b4')~\eta_i=1-\xi_i=0,\beta_i=\alpha_3=\alpha_4=0.
\end{array}
\end{equation*}
Else if $\beta_i=1$, it acquires $\alpha_1\xi_i+\alpha_2(1-\xi_i)=1-\xi_i$, which derives
\begin{equation*}\begin{array}{ll}
(b1'')~\eta_i=1-\xi_i,\alpha_1=0,\beta_i=\alpha_2=1; \\ (b2'')~\eta_i=1-\xi_i=\frac{1}{2},\beta_i=\alpha_1=1,\alpha_2=0;\\
(b3'')~\eta_i=1-\xi_i=1,\beta_i=\alpha_1=\alpha_2=1;\\
(b4'')~\eta_i=1-\xi_i=0,\beta_i=1,\alpha_1=\alpha_2=0.
\end{array}
\end{equation*}

Similarly, for case $(c)$, if $\beta_i=0$, it gains $\alpha_3\xi_i+\alpha_4(1-\xi_i)=0~\mathrm{or}~1$, that is,
\begin{equation*}\begin{array}{ll}
(c1')~\eta_i=0,\xi_i\in \mathbb{B}, \beta_i=\alpha_3=\alpha_4=0;~\mathrm{or}\\ (c2')~\eta_i=1,\xi_i\in\mathbb{B},\beta_i=0,\alpha_3=\alpha_4=1.
\end{array}
\end{equation*}
Else if $\beta_i=1$, it gains $\alpha_1\xi_i+\alpha_2(1-\xi_i)=0~\mathrm{or}~1$, that is,
\begin{equation*}\begin{array}{ll}
(c1'')~\eta_i=0,\xi_i\in \mathbb{B}, \beta_i=1,\alpha_1=\alpha_2=0; ~\mathrm{or}\\ (c2'')~\eta_i=1,\xi_i\in\mathbb{B},\beta_i=1,\alpha_1=\alpha_2=1.
\end{array}
\end{equation*}

The above solutions can be classified into four groups:
\begin{description}
  \item[\textbf{Group 1:}]~~~~$\eta_i=\xi_i$; $\eta_i=\xi_i=0$; $\eta_i=1-\xi_i=\frac{1}{2}$.

  \item[Note 1:]~$\eta_i=\xi_i\Leftrightarrow$ $(a1')$ or $(a1'')$; $\eta_i=\xi_i=0\Leftrightarrow$ $(a2')$ or $(a2'')$;
  $\eta_i=1-\xi_i=\frac{1}{2}\Leftrightarrow$ $(b2')$ or $(b2'')$.

  \item[\textbf{Group 2:}]~~~~$\eta_i=\xi_i=1$; $\eta_i=1-\xi_i=1$; $\eta_i=1,\xi_i\in \mathbb{B}$.
  \item[Note 2:]~$\eta_i=\xi_i=1\Leftrightarrow$ $(a3')$ or $(a3'')$;
  $\eta_i=1-\xi_i=1\Leftrightarrow$ $(b3')$ or $(b3'')$;
  $\eta_i=1,\xi_i\in \mathbb{B}\Leftrightarrow$ $(c2')$ or $(c2'')$.

  \item[\textbf{Group 3:}]~~~~$\eta_i=\xi_i=\frac{1}{2}$; $\eta_i=1-\xi_i$; $\eta_i=1-\xi_i=1$ (occurred in Group 2).
  \item[Note 3:]~$\eta_i=\xi_i=\frac{1}{2}\Leftrightarrow$ $(a4')$ or $(a4'')$;
  $\eta_i=1-\xi_i\Leftrightarrow$ $(b1')$ or $(b1'')$.

  \item[\textbf{Group 4:}]~~~~$\eta_i=\xi_i=0$ (occurred in Group 1), $\eta_i=1-\xi_i=0$, and $\eta_i=0,\xi_i\in \mathbb{B}$.
  \item[Note 4:]~~$\eta_i=1-\xi_i=0\Leftrightarrow$ $(b4')$ or $(b4'')$;
  $\eta_i=0,\xi_i\in \mathbb{B}\Leftrightarrow$ $(c1')$ or $(c1'')$.
\end{description}
Particularly, $\eta_i=\xi_i\in \{0,1,\frac{1}{2}\}$ and $\eta_i=1-\xi_i=\frac{1}{2}$ can be viewed as $\eta_i=\xi_i$ belonging to \textbf{Group 1}, while $\eta_i=1-\xi_i\in \{0,1,\frac{1}{2}\}$ and $\eta_i=\xi_i=\frac{1}{2}$ can be regarded as $\eta_i=1-\xi_i$ belonging to \textbf{Group 3}. Then, the above four groups respectively correspond to $\Omega_1,\Omega_2,\Omega_3$ and $\Omega_4$ defined in (\ref{theorem-condition}).
Besides, the cases in the same group can be handled by the same pinning controllers. More precisely, only when the solution of equations (\ref{equations-BCN-proof}) can be covered by the cases within at most two groups, can we obtain effective solutions.

(Necessity.) Suppose that matrix $\mathbf{T}_{\gamma_\sigma}$ is covered by at most two sets, then we prove that equations (\ref{equations-BCN-proof}) is solvable.

Case 1: $\mathrm{Col}(\mathbf{T}_{\gamma_\sigma})$ is contained by one of sets $\Omega_1,\cdots,\Omega_4$. If part of $\mathrm{Col}_i(\mathbf{T}_{\gamma_\sigma}), i\in [1,2^{|\mathcal{N}_{\gamma_\sigma}|}]$ is equal to $\mathrm{Col}_i(\mathbf{L}_{\gamma_\sigma})$, the other part is equal to $1-\mathrm{Col}_i(\mathbf{L}_{\gamma_\sigma})=\frac{1}{2}$, it derives
$M_{\odot_{\gamma_\sigma}}=\delta_2[\ast,\ast,1,2],\widehat{\Psi}_{\gamma_\sigma}=\delta_2[\underbrace{2,\cdots,2}_{2^{|\mathcal{N}_{\gamma_\sigma}|}}]$; or $M_{\odot_{\gamma_\sigma}}=\delta_2[1,2,\ast,\ast],\widehat{\Psi}_{\gamma_\sigma}=\delta_2[\underbrace{1,\cdots,1}_{2^{|\mathcal{N}_{\gamma_\sigma}|}}]$,
where the elements in the positions of `$\ast$' can be arbitrary chosen from $\{0,1\}$. Besides, other cases can similarly derive effective pinning controllers ($M_{\odot_{\gamma_\sigma}}$ and $\widehat{\Psi}_{\gamma_\sigma}$).

Case 2: $\mathrm{Col}(\mathbf{T}_{\gamma_\sigma})$ is contained by arbitrary two of sets $\Omega_1,\cdots,\Omega_4$. Next, we verify one situation and the others can be similarly proved.
Let $\Theta\subset [1,2^{|\mathcal{N}_{\gamma_\sigma}|}]$ be the index set of columns $\mathrm{Col}_i(\mathbf{T}_{\gamma_\sigma})$ belonging to $\Omega_1$, then the other columns satisfy $\mathrm{Col}_i(\mathbf{T}_{\gamma_\sigma})\in \Omega_2,i\in [1,2^{|\mathcal{N}_{\gamma_\sigma}|}]\backslash \Theta$. In this case, the structure matrices can be determined as
\begin{equation}\left\{\small{\begin{array}{l}\label{eg1}
M_{\odot_{\gamma_\sigma}}=\delta_2[1,1,1,2],\\
\widehat{\Psi}_{\gamma_\sigma}=\delta_2[\alpha_1,\alpha_2,\cdots,\alpha_{2^{|\mathcal{N}_{\gamma_\sigma}|}}],\\
\alpha_i=2~\mathrm{if}~i\in \Theta,\\
\alpha_i=1~\mathrm{if}~i\in [1,2^{|\mathcal{N}_{\gamma_\sigma}|}]\backslash\Theta,\\
\end{array}}\right.
\end{equation}
or
\begin{equation}\left\{\small{\begin{array}{l}\label{eg2}
M_{\odot_{\gamma_\sigma}}=\delta_2[1,2,1,1],\\
\widehat{\Psi}_{\gamma_\sigma}=\delta_2[\alpha_1,\alpha_2,\cdots,\alpha_{2^{|\mathcal{N}_{\gamma_\sigma}|}}],\\
\alpha_i=1~\mathrm{if}~i\in \Theta,\\
\alpha_i=2~\mathrm{if}~i\in [1,2^{|\mathcal{N}_{\gamma_\sigma}|}]\backslash\Theta.\\
\end{array}}\right.
\end{equation}
Taking (\ref{eg1}) or (\ref{eg2}) into equations (\ref{equations-BCN-proof}), it obviously holds. Therefore, the proof is completed.
\end{proof}

\subsection{Design Non-Uniform State Feedback Controllers}
For each pinning node $x_{\gamma_\sigma},\sigma\in [1,\tau]$, designing specific state feedback controller for each possibility is apparently feasible, which has been proved in \cite{liff2015pinning} and \cite{liff2019set}. However, for the sake of saving control cost, we preferentially choose uniform state feedback controllers if equations (\ref{stable-structure-matrix}) is solvable.

Sequentially, examining the solvability criterion in Theorem \ref{theorem-2-condition}, and collecting the subscript indices of pinning nodes, whose equations (\ref{stable-structure-matrix}) is unsolvable, into set $\Lambda_2\subseteq \Lambda$. Besides, denote $\Lambda_1=\Lambda\backslash \Lambda_2$ as the solvable index set. The non-uniform state feedback controller imposed on $x_{\gamma_\varrho}\in\Lambda_2$ depends on each possibility $f_{\gamma_\varrho}^\kappa,\kappa\in[1,l_{\gamma_\varrho}]$, and is given as that:
\begin{equation}\label{state-feedback-pinning-controller-1}
u^\kappa_{\gamma_\varrho}(t)=\varphi^\kappa_{\gamma_\varrho}([x_j(t)]_{j\in \mathcal{N}_{\gamma_\varrho}}).
\end{equation}
After that, the state updating of each node $x_i,i\in[1,n]$ obeys the following rule:
\begin{equation*}
x_i(t+1)=g_i([x_j(t)]_{j\in \mathcal{N}^\circ_i}),
\end{equation*}
where $g_i\in \{g^\kappa_i([x_j(t)]_{j\in \mathcal{N}^\circ_i})|\kappa\in[1,l_i]\}$, and
\begin{equation}g^\kappa_i=\left\{\begin{array}{ll}\label{g}
f^\kappa_i([x_j(t)]_{j\in \mathcal{N}_i}), &i\in[1,n]\backslash \Lambda,\\
u_i(t) \odot_i f^\kappa_i([x_j(t)]_{j\in \mathcal{N}_i}), &i\in \Lambda_1,\\
u^\kappa_i(t) \odot^\kappa_i f^\kappa_i([x_j(t)]_{j\in \mathcal{N}_i}), &i\in \Lambda_2.
\end{array}\right.
\end{equation}
Thereinto, $\odot^\kappa_i:\mathscr{D}^2\rightarrow \mathscr{D}$ is logical function connecting control input $u^\kappa_i(t)$ and original dynamic equation $f^\kappa_i([x_j(t)]_{j\in \mathcal{N}_i})$.
Moreover, denote $\widehat{\Psi}^\kappa_i\in \mathscr{L}^{2\times 2^{|\mathcal{N}_{i,\kappa}|}}$ and $M^\kappa_{\odot_i}\in \mathscr{L}^{2\times 4}$ as the structure matrix of $u^\kappa_i$ and $\odot^\kappa_i$.

Similarly, for each pinning node $x_{\gamma_\varrho}\in\Lambda_2$, the state feedback pinning controllers can be designed by solving structure matrices $\widehat{\Psi}^\kappa_{\gamma_\varrho}$ and $M^\kappa_{\odot_{\gamma_\varrho}},\kappa\in [1,l_{\gamma_\varrho}]$ from the below equations:
\begin{equation}\label{stable-structure-matrices}\begin{aligned}
M^\kappa_{\odot_{\gamma_\varrho}}\widehat{\Psi}^\kappa_{\gamma_\varrho}(I_{2^{|\mathcal{N}_{\gamma_\varrho,\kappa}|}}\otimes \mathbf{L}^\kappa_{\gamma_\varrho}) \Phi_{|\mathcal{N}_{\gamma_\varrho,\kappa}|}&=G^\kappa_{\gamma_\varrho}(I_{2^{|\mathcal{N}^\circ_{\gamma_\varrho,\kappa}|}}\otimes \mathbf{1}^\top_{2^{|\mathcal{N}^\star_{\gamma_\varrho,\kappa}|}})\\
&:=\mathbf{T}^\kappa_{\gamma_\varrho},
\end{aligned}
\end{equation}
where $\mathbf{L}^\kappa_{\gamma_\varrho}=F_{\gamma_\varrho}^\kappa\mathbf{W}_{\gamma_\varrho}$, $G^\kappa_{\gamma_\varrho} \in \mathscr{L}^{2\times 2^{|\mathcal{N}^\circ_{\gamma_\varrho,\kappa}|}}$, and $\mathbf{T}^\kappa_{\gamma_\varrho}\in \mathscr{L}^{2\times 2^{|\mathcal{N}_{\gamma_\varrho,\kappa}|}}$. Fortunately, equations (\ref{stable-structure-matrices}) is always solvable, whose detailed proof can be found in \cite{liff2015pinning,liff2019set}.
Subsequently, denote the extended structure matrix of $G^\kappa_{\gamma_\rho}\in \mathscr{L}^{2\times 2^{|\mathcal{N}^\circ_{\gamma_\varrho,\kappa}|}}$ by $\widehat{G}^\kappa_{\gamma_\rho}\in \mathscr{L}^{2\times 2^{|\mathcal{N}^\circ_{\gamma_\rho}|}}$ satisfying $\widehat{G}^\kappa_{\gamma_\rho}\ltimes_{j\in \mathcal{N}^\circ_{\gamma_\rho}}\mathbf{x}_j(t)=G^\kappa_{\gamma_\rho}\ltimes_{j\in \mathcal{N}^\circ_{\gamma_\rho,\kappa}}\mathbf{x}_j(t)$.

Eventually, the mathematical expectation of pinning controlled PBN is derived as that:
\begin{equation}
\mathbb{E}\{\mathbf{x}_i(t+1)\}=G_i\mathbb{E}\{\ltimes_{j\in \mathcal{N}^\circ_i}\mathbf{x}_j(t)\},
\end{equation}
where $G_i=\widehat{F}_i$ if $i\in [1,n]\backslash\Lambda$; $G_i=\sum_{\kappa=1}^{l_i}p^\kappa_iG^\kappa_i$ if $i\in \Lambda_1$; and $G_i=\sum_{\kappa=1}^{l_i}p^\kappa_i\widehat{G}^\kappa_i$ if $i\in \Lambda_2$.

%

\section{Pinning Control for Stabilizing PBNs to A certain State}\label{section-graphicaltheory}
Notice that if pinning controlled PBN (\ref{equation-BCN}) has achieved stabilization, there may occur several different steady states. Hence, in this section, we further impose another pinning control on PBN (\ref{equation-BCN}) to achieve stabilized at a unique and specific steady state.
In the following sequel, let steady state be $\{\epsilon_1,\epsilon_2,\cdots,\epsilon_n\},\epsilon_i\in \mathscr{D},i\in [1,n]$ and its algebraic form be $\delta_{2^n}^\epsilon=\ltimes_{i=1}^n \delta_2^{2-\epsilon_i}$.

Thereupon, we find matrices $\mathbf{Q}_i,\in\mathscr{L}^{2\times 2^{|\mathcal{N}^\circ_i|}},i\in[1,n]$, and binary variables $\lambda_1,\cdots,\lambda_n\in \mathscr{D}$ such that
\begin{equation}\label{objective-function}
\mathrm{min}~~~~ \Xi\triangleq\sum_{i=1}^n \lambda_i,
\end{equation}
subjects to
\begin{equation}\left\{\begin{array}{c}\label{constraint}
\delta_2^{2-\epsilon_1}=[(1-\lambda_1)G_1+\lambda_1\mathbf{Q}_1]\ltimes_{j\in \mathcal{N}^\circ_1}\delta_2^{2-\epsilon_j},\\
\vdots\\
\delta_2^{2-\epsilon_n}=[(1-\lambda_n)G_n+\lambda_n\mathbf{Q}_n]\ltimes_{j\in \mathcal{N}^\circ_n}\delta_2^{2-\epsilon_j}.
\end{array}\right.
\end{equation}
\begin{remark}
Here, $\Xi$ is termed as control cost function, and objective function (\ref{objective-function}) determines
the number of pinning nodes as minimal as possible. This optimization problem can be deemed as a linear programming problem and can be solved by existing methods.
\end{remark}

Assuming that the solution of the above optimization problem is $\Xi=\pi$ with $\lambda_{\zeta_1}=\lambda_{\zeta_2}=\cdots=\lambda_{\zeta_\pi}=1$, then the second group of pinning nodes (to achieve stabilizing at a certain state) is $x_{\zeta_1},\cdots,x_{\zeta_\pi}$, whose subscript indices are collected into set $\Gamma$.
The state feedback controller imposed on node $x_{\zeta_\iota},\iota\in[1,\pi]$ is described as:
\begin{equation}\left\{\begin{array}{l}\label{equation-BCN-2}
x_{\zeta_\iota}(t+1)=v_{\zeta_\iota}(t) \oplus_{\zeta_\iota} g_{\zeta_\iota}([x_j(t)]_{j\in \mathcal{N}^\circ_{\zeta_\iota}}),\\
v_{\zeta_\iota}(t)=\phi_{\zeta_\iota}([x_j(t)]_{j\in \mathcal{N}^\circ_{\zeta_\iota}}),
\end{array}\right.
\end{equation}
where $\phi_{\zeta_\iota}([x_j(t)]_{j\in \mathcal{N}^\circ_{\zeta_\iota}})$ is logical function decided by node $x_{\zeta_\iota}$'s in-neighbors, and $\oplus_{\zeta_\iota}: \mathscr{D}^2\rightarrow \mathscr{D}$ is logical function connecting state feedback controller $v_{\zeta_\iota}$ and $g_{\zeta_\iota}([x_j(t)]_{j\in \mathcal{N}^\circ_{\zeta_\iota}})$.

Likewise, in resorting to STP, the structure matrices of $\phi_{\zeta_\iota}$ and
$\oplus_{\zeta_\iota},\iota\in [1,\pi]$ are respectively denoted by $\Upsilon_\iota\in \mathscr{L}^{2\times2^{|\mathcal{N}^\circ_\iota|}}$ and $M_{\oplus_\iota}\in \mathscr{L}^{2\times 4}$. Then we can obtain the algebraic form of (\ref{equation-BCN-2}) as below:
\begin{equation}\left\{\small{\begin{array}{l}
\mathbb{E}\{\mathbf{x}_{\zeta_\iota}(t+1)\}=M_{\oplus_{\zeta_\iota}}\Upsilon_{\zeta_\iota}(I_{2^{|\mathcal{N}^\circ_{\zeta_\iota}|}}\otimes G_{\zeta_\iota})\Phi_{|\mathcal{N}^\circ_{\zeta_\iota}|}\mathbb{E}\{\ltimes_{j\in\mathcal{N}^\circ_{\zeta_\iota}}\mathbf{x}_j(t)\},\\
v_{\zeta_\iota}(t)=\Upsilon_{\zeta_\iota}\ltimes_{j\in\mathcal{N}^\circ_{\zeta_\iota}}\mathbf{x}_j(t),
\end{array}}\right.
\end{equation}
where $\Upsilon_{\zeta_\iota}$ and $M_{\oplus_{\zeta_\iota}}$ can be solved from equations:
\begin{equation}\label{stable-structure-matrix-2}
M_{\oplus_{\zeta_\iota}}\Upsilon_{\zeta_\iota}(I_{2^{|\mathcal{N}^\circ_{\zeta_\iota}|}}\otimes G_{\zeta_\iota})\Phi_{|\mathcal{N}^\circ_{\zeta_\iota}|}=\mathbf{Q}_{\zeta_\iota}.
\end{equation}

\begin{remark}
The solvability of equations (\ref{stable-structure-matrix-2}) is guaranteed. Since each $\delta_2^{2-\epsilon_i}, i\in [1,n]$ is equal to $\delta_2^1$ or $\delta_2^2$, its corresponding matrix $\mathbf{Q}_i, i\in [1,n]$ can be directly valued as
$\delta_2[\underbrace{1,\cdots,1}_{2^{|\mathcal{N}^\circ_i|}}]$ or $\delta_2[\underbrace{2,\cdots,2}_{2^{|\mathcal{N}^\circ_i|}}]$. It implies that the column set of each $\mathbf{Q}_i, i\in [1,n]$ can be completely covered by $\Omega_2$ or $\Omega_4$. Therefore, equation (\ref{stable-structure-matrix-2}) must be solvable according to Theorem \ref{theorem-2-condition}.
\end{remark}

Here, the pinning controlled PBN (\ref{equation-PBN}) is represented as that:
\begin{equation}x_i(t+1)=\left\{\begin{array}{l}\label{equation-BCN-final}
v_i(t) \oplus_i [u_i(t)\odot_i f_i([x_j(t)]_{j\in \mathcal{N}_i})], ~i\in\Lambda\cap \Gamma, \\
u_i(t)\odot_i f_i([x_j(t)]_{j\in \mathcal{N}_i}), ~~~~~~~~~i\in\Lambda\backslash\Gamma, \\
v_i(t) \oplus_i f_i([x_j(t)]_{j\in \mathcal{N}_i}),~~~~~~~~~i\in\Gamma\backslash\Lambda,\\
f_i([x_j(t)]_{j\in \mathcal{N}_i}),~~~~~~~~~~~~~i\in[1,n]\backslash(\Lambda\cap\Gamma).
\end{array}\right.
\end{equation}
Thereinto, $u_i,\odot_i, i\in \Lambda$ are obtained by solving equations (\ref{stable-structure-matrix}) or (\ref{stable-structure-matrices}),
and $v_i,\oplus_i, i\in \Gamma$ are derived by solving (\ref{stable-structure-matrix-2}).
To conclude, the procedure of designing pinning control for stabilizing PBN (\ref{equation-PBN}) to a prescribed state $\delta_{2^n}^\epsilon$ is established by the following steps:
\begin{enumerate}
  \item Deduce wiring digraph $\mathbf{G}=(\mathbf{V},\mathbf{E})$ of PBN (\ref{equation-PBN}). If $\mathbf{G}$ is acyclic, then goes to step 6, otherwise goes to next step.
  \item Find the fixed points and cycles existing in $\mathbf{G}$ by the depth-first search algorithm, thereupon determine a possible FAS $\{e_1,\cdots,e_\kappa\}$ and collect ending nodes $x_{\gamma_1},\cdots,x_{\gamma_\tau}$ into $\Lambda$.
  \item For each pinning node $x_{\gamma_\sigma},\sigma\in [1,\tau]$, if equations (\ref{stable-structure-matrix}) is solvable, then do the first operation; otherwise let $\Lambda_2\leftarrow\Lambda_2 \cup \{\gamma_\sigma\}$ and do the second operation:
      \begin{itemize}
        \item Solve structure matrices $M_{\odot_{\gamma_\sigma}}$ and $\widehat{\Psi}_{\gamma_\sigma}$ from (\ref{stable-structure-matrix}). (*Impose uniform controller (\ref{state-feedback-pinning-controllers}) to $x_{\gamma_\sigma}\in \Lambda_1$*)
        \item Solve a series of structure matrices $M^\kappa_{\odot_{\gamma_\sigma}}$ and $\widehat{\Psi}^\kappa_{\gamma_\sigma},\kappa\in[1,l_{\gamma_\sigma}]$ from (\ref{stable-structure-matrices}). (*Impose non-uniform controller (\ref{state-feedback-pinning-controller-1}) to each $f^\kappa_{\gamma_\sigma},\kappa\in[1,l_{\gamma_\sigma}]$ of $x_{\gamma_\sigma}\in \Lambda_2$*)
      \end{itemize}
  \item Solve optimization objective (\ref{objective-function}) subjecting to (\ref{constraint}), and derive $\Xi=\pi$ and matrices $\mathbf{Q}_{\zeta_1},\cdots,\mathbf{Q}_{\zeta_\pi}$ with $\lambda_{\zeta_1}=\cdots=\lambda_{\zeta_\pi}=1$. Then, collect pinning nodes $x_{\zeta_1},\cdots,x_{\zeta_\pi}$ into $\Gamma$.
  \item For node $x_{\zeta_\iota},\iota\in[1,\pi]$, solve structure matrices $M_{\oplus_{\zeta_\iota}}$ and $\Upsilon_{\zeta_\iota}$ from (\ref{stable-structure-matrix-2}).
(*Impose uniform controllers (\ref{equation-BCN-2}) to pinning nodes $x_{\zeta_1}\in \Gamma$*)
  \item The solved state feedback controllers are respectively imposed on the different kinds of pinning nodes as (\ref{equation-BCN-final}).
\end{enumerate}

\begin{remark}
Compared with traditional pinning control design in \cite{wanglq2019pinning} and \cite{liff2019set}, our control design strategy is more practical and can be implemented to sparse large-scale networks. Besides, the time complexity is reduced from $O(2^{2n})$ to $O(n2^{\max_{i\in[1,n]}|\mathcal{N}_i|})$.
\end{remark}

\section{Simulation}\label{section-stabilization}
In this section, a reduced mammalian cell cycle network \cite{faure2006dynamical} is performed to demonstrate the design process of our pinning control. In this setup, the mammalian cell-cycle encountering a mutated phenotype is postulated as a PBN, which has nine variables: $x_1,x_2,\cdots$ and $x_9$ (respectively representing genes Rb, E2F, CycE, CycA, p27, Cdc 20, Cdh 1, UbcH 10, and CycB). Particularly, this mammalian cell cycle may encounter mutation with probability 0.01, and the mutated mammalian cell cycle network has two possible constituent forms with the same probability 0.005. The form of the mutation depends on the expression status of external input CycD \cite{faryabi2008regulatory}.

To be specific, the logical dynamics of normal mammalian cell cycle network is described as:
\begin{equation}\left\{\begin{aligned}\label{ex-1}
f_1^1(\ast)=&(\neg x_3\wedge \neg x_4 \wedge \neg x_9)\vee (x_5 \wedge \neg x_9),\\
f_2^1(\ast)=&(\neg x_1\wedge \neg x_4 \wedge \neg x_9)\vee (x_5 \wedge \neg x_1\wedge \neg x_9),\\
f_3^1(\ast)=&x_2\wedge \neg x_1,\\
f_4^1(\ast)=&[x_2 \wedge\neg x_1\wedge \neg x_6 \wedge \neg (x_7\wedge x_8)]\vee \\
&~~[\neg x_1 \wedge \neg x_6 \wedge x_4 \wedge \neg (x_7\wedge x_8)],\\
f_5^1(\ast)=&(\neg x_3\wedge \neg x_4 \wedge \neg x_9)\vee [x_5 \wedge \neg (x_3\wedge x_4)\wedge \neg x_9],\\
f_6^1(\ast)=&x_9,\\
f_7^1(\ast)=&(\neg x_4\wedge \neg x_9)\vee x_6 \vee (x_5 \wedge \neg x_9),\\
f_8^1(\ast)=&\neg x_7 \vee [x_7 \wedge x_8 \wedge (x_6\vee x_4 \vee x_9)], \\
f_9^1(\ast)=&\neg x_6 \wedge \neg x_7.
\end{aligned}\right.
\end{equation}
Moreover, when the value of CycD is 0, the mutated mammalian cell cycle network is performed as form A:
\begin{equation}\left\{\begin{array}{l}\label{ex-2}
f^{2}_1(\ast)=\neg x_3\wedge \neg x_4 \wedge \neg x_9,~~f^{2}_2(\ast)=\neg x_1\wedge \neg x_4 \wedge \neg x_9,\\
f^{2}_3(\ast)=f_3(\ast),~~f^{2}_4(\ast)=f_4(\ast),~~f_5^{2}(\ast)=f_5(\ast),\\
f_6^{2}(\ast)=f_6(\ast),~~f_7^{2}(\ast)=(\neg x_4\wedge \neg x_9)\vee x_6,\\
f_8^{2}(\ast)=f_8(\ast),~~f_9^{2}(\ast)=f_9(\ast);
\end{array}\right.
\end{equation}
and when the value of CycD is 1, it is performed as form B:
\begin{equation}\left\{\begin{array}{l}\label{ex-3}
f^{3}_1(\ast)=0,~~~~~~~f^{3}_2(\ast)=f^{2}_2(\ast),~~f^{3}_3(\ast)=f^{2}_3(\ast),\\
f^{3}_4(\ast)=f^{2}_4(\ast),~~f_5^{3}(\ast)=f^{2}_5(\ast),~~f_6^{3}(\ast)=f^{2}_6(\ast),\\
f_7^{3}(\ast)=f^{2}_7(\ast),~~f_8^{3}(\ast)=f^{2}_8(\ast),~~f_9^{3}(\ast)=f^{2}_9(\ast).
\end{array}\right.
\end{equation}
As indicated in \cite{faryabi2008regulatory}, gene-activity profile would be forced to the desirable state $(1,1,1,0,0,0,1,0,0)^\top$ equivalent to $\delta_{512}^{44}$. Thereby, in the following, we execute pinning control to stabilize this PBN to $\delta_{512}^{44}$.

\textbf{Step 1:} the wiring digraph $\mathbf{G}$ of PBN (\ref{equation-PBN}) defined by (\ref{ex-1}), (\ref{ex-2}) and (\ref{ex-3}) is established in Fig. \ref{fig1}, and it has several fixed points and cycles. Then proceed to \textbf{Step 2}, one obtains a feasible FAS as $\{x_3\rightarrow x_1,x_1\rightarrow x_4,x_2\rightarrow x_4,x_4\rightarrow x_4,x_7\rightarrow x_4,x_8\rightarrow x_4,x_3\rightarrow x_5, x_5\rightarrow x_5,x_9\rightarrow x_6,x_8\rightarrow x_8,x_7\rightarrow x_9\}$, then the first group of pinning nodes are $x_1,x_4,x_5,x_6,x_8$ and $x_9$.

\begin{figure}[h!]
\centering
\includegraphics[width=0.27\textwidth=0.2]{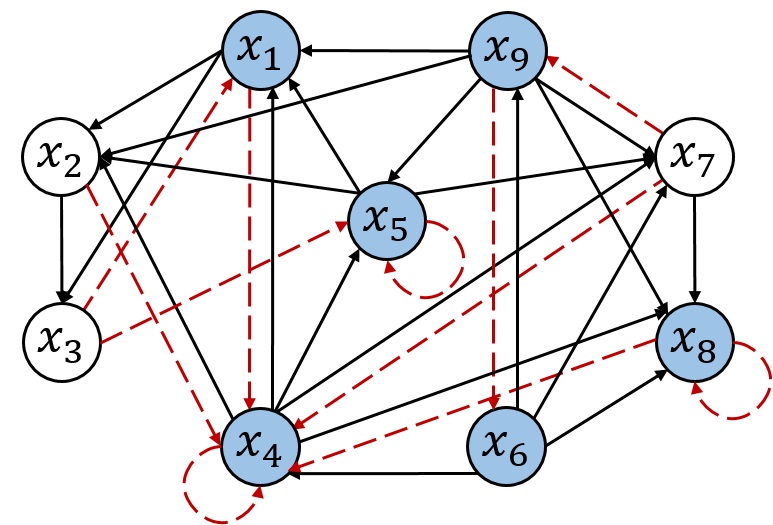}
\caption{The wiring diagrams of PBN (\ref{equation-PBN}) defined in the section \ref{section-stabilization}. In this figure, each directed edge $i\rightarrow j,i,j\in[1,9]$ expresses $x_i$ being the functional variable of $f^1_j$ or $f^2_j$ or $f^3_j$. Besides, red edges are the edges in the FAS; and blue nodes are the selected pinning nodes.}\label{fig1}
\end{figure}
Before operating \textbf{Step 3}, we denote the updating matrix of gene $x_i,i\in[1,9]$ by $\textbf{L}_i\in \mathscr{P}^{2\times 2^{|\mathcal{N}_i|}}$, which is easily derived by (\ref{Ex(t+1)}).
Besides, the structure matrices after deleting edges in FAS are obtained as follows:
\begin{equation*}G_1=\left(\begin{array}{cccccccc}
0&0.99&0&0&0.995&0.995&0&0.995 \\
1&0.01&1&1&1&0.005&1&0.005
\end{array}\right),
\end{equation*}
\begin{equation*}G_2=\textbf{L}_2=\left(\begin{array}{c}
0~\cdots~0~0.99~0~0~0~1~0~1 \\
1~\cdots~1~0.01~1~1~1~0~1~0
\end{array}\right)\in \mathscr{P}^{2\times 16},
\end{equation*}
$G_3=\textbf{L}_3=\delta_2[2,1,2,2]$, $G_4=\delta_2[2,1]$, $G_5=\delta_2[2,2,2,1]$, $G_6=\delta_2^2$, \begin{equation*}G_7=\textbf{L}_7=\left(\begin{array}{c}
1~1~0~0.99~1~1~0~0~1~1~0~1~1~1~0~1 \\
0~0~1~0.01~0~0~1~1~0~0~1~0~0~0~1~0
\end{array}\right),
\end{equation*}
$G_8=\delta_2[1,\cdots,1,2,1,1]\in \mathscr{P}^{2\times 16}$ and $G_9=\delta_2[2,1]$.
Meanwhile, according to equation (\ref{stable-struture-matrix-0}), it derives that:
\begin{equation*}\mathbf{T}_1=\left(\small{\begin{array}{c}
0~0~0.99~0.99~\underbrace{0~\cdots~0}_6~0.995~0.995~0~0~0.995~0.995 \\
1~1~0.01~0.01~\underbrace{1~\cdots~1}_6~0.005~0.005~1~1~0.005~0.005
\end{array}}\right),
\end{equation*}
$\mathbf{T}_4=\delta_2[\underbrace{2,\cdots,2}_{32},\underbrace{1,\cdots,1}_{32}]$, $\mathbf{T}_5=\delta_2[\underbrace{2,\cdots,2}_{12},\underbrace{1,\cdots,1}_4]$, $\mathbf{T}_6=\delta_2[2,2]$,
$\mathbf{T}_8=\delta_2[\underbrace{1,\cdots,1}_{26},2,2,\underbrace{1,\cdots,1}_4]$, and $\mathbf{T}_9=\delta_2[2,2,1,1]$.
Then, we can examine the solvability of equations (\ref{stable-structure-matrix}) for pinning nodes $x_i,i\in \{1,4,5,6,8,9\}$ according to Theorem \ref{theorem-2-condition}, it reveals that $\mathbf{T}_1=M_{\odot_1}\widehat{\Psi}_1(I_{16}\otimes \mathbf{L}_1) \Phi_4\in\mathscr{P}^{2\times 16}$ is unsolvable, while $\mathbf{T}_4=M_{\odot_4}\widehat{\Psi}_4(I_{64}\otimes \mathbf{L}_4) \Phi_6$,
$\mathbf{T}_5=M_{\odot_5}\widehat{\Psi}_5(I_{16}\otimes \mathbf{L}_5)\Phi_4$,
$\mathbf{T}_6=M_{\odot_6}\widehat{\Psi}_6(I_{2}\otimes \mathbf{L}_6) \Phi_1$,
$\mathbf{T}_8=M_{\odot_8}\widehat{\Psi}_8(I_{32}\otimes \mathbf{L}_8)\Phi_5$ and
$\mathbf{T}_9=M_{\odot_9}\widehat{\Psi}_9(I_{4}\otimes \mathbf{L}_9) \Phi_2$ are solvable.
Thereby, with respect to pinning node $x_1$, we obtain structure matrices (\ref{stable-structure-matrices}) as $\mathbf{T}^1_1=\delta_2[2,2,1,1,2,\cdots,2,1,1,2,2,1,1]$, $\mathbf{T}^2_1=\delta_2[2,\cdots,2,1,1,2,2,1,1]$ and $\mathbf{T}^{3}_1=\delta_2[2,\cdots,2]$. Then, the state feedback pinning controllers imposed on $f^1_1$, $f^2_1$ and $f^3_1$ can be respectively solved from $\mathbf{T}^1_1=M^1_{\odot_1}\widehat{\Psi}^1_1(I_{16}\otimes \mathbf{L}^1_1) \Phi_4$, $\mathbf{T}^2_1=M^2_{\odot_1}\widehat{\Psi}^2_1(I_{16}\otimes \mathbf{L}^2_1) \Phi_4$, and $\mathbf{T}^{3}_1=M^{3}_{\odot_1}\widehat{\Psi}^{3}_1(I_{16}\otimes \mathbf{L}^{3}_1) \Phi_4$. It derives that
$$M^1_{\odot_1}=M^2_{\odot_1}=M^{3}_{\odot_1}=\delta_2[1,1,1,2],\widehat{\Psi}^1_1=\delta_2[\underbrace{2,\cdots,2}_{14},1,1],$$
$$\widehat{\Psi}^2_1=\delta_2[\underbrace{2,\cdots,2}_{10},1,2,2,2,1,2],~\mathrm{and}~\widehat{\Psi}^{3}_1=\delta_2[\underbrace{2,\cdots,2}_{16}].$$
Simultaneously, based on equation (\ref{stable-structure-matrix}),
the efficient state feedback controllers applied on pinning genes $x_4,x_5,x_6,x_8,x_9$ can be respectively designed as $$M_{\odot_4}=\delta_2[1,1,2,2],\widehat{\Psi}_4=\delta_2[\underbrace{2,\cdots,2}_{32},\underbrace{1,\cdots,1}_{32}];$$ $$M_{\odot_5}=\delta_2[1,1,2,2],\widehat{\Psi}_5=\delta_2[\underbrace{2,\cdots,2}_{12},\underbrace{1,\cdots,1}_{4}];$$
$$M_{\odot_6}=\delta_2[2,2,2,2],\widehat{\Psi}_6=\delta_2[1,1];$$ $$M_{\odot_8}=\delta_2[1,1,2,2],\widehat{\Psi}_8=\delta_2[\underbrace{1,\cdots,1}_{26},2,2,\underbrace{1,\cdots,1}_4];$$ $$\mathrm{and}~M_{\odot_9}=[1,1,2,2], \widehat{\Psi}_9=[2,2,1,1].$$
After that, the wiring digraph of this pinning controlled PBN is acyclic, and all possible models achieve globally stabilized.

To proceed, perform \textbf{Step 4} on the stabilized PBN. The solution of optimization problem (\ref{objective-function}) is $\Xi=8$ and $\lambda_1=\lambda_2=\lambda_3=\lambda_4=\lambda_5=\lambda_7=\lambda_8=\lambda_9=1$. Subsequently, we obtain $\mathbf{Q}_1=\delta_2[1,\cdots,1]\in\mathscr{L}^{2\times 8}$ $\mathbf{Q}_2=\mathbf{Q}_7=\delta_2[1,\cdots,1]\in\mathscr{L}^{2\times 16}$, $\mathbf{Q}_3=\mathbf{Q}_5=\delta_2[1,\cdots,1]\in\mathscr{L}^{2\times 4}$, $\mathbf{Q}_4=\mathbf{Q}_9=\delta_2[0,0]$, $\mathbf{Q}_8=\delta_2[2,\cdots,2]\in\mathscr{L}^{2\times 16}$.
Then, go to \textbf{Step 5} and obtain
$$M_{\oplus_1}=M_{\oplus_2}=M_{\oplus_3}=M_{\oplus_5}=M_{\oplus_7}=\delta_2[1,1,1,1],$$ $$M_{\oplus_4}=M_{\oplus_8}=M_{\oplus_9}=\delta_2[2,2,2,2],$$
$$\Upsilon_1=\delta_2[1,\cdots,1]\in\mathscr{L}^{2\times8},\Upsilon_2=\Upsilon_7=\Upsilon_8=\delta_2[1,\cdots,1]\in\mathscr{L}^{2\times16},$$ $$\Upsilon_3=\Upsilon_5=\delta_2[1,1,1,1],~\mathrm{and}~\Upsilon_4=\Upsilon_9=\delta_2[1,1].$$
Finally, the pinning node sets are designed as $\Lambda\cap\Gamma=\{1,4,5,8,9\}$, $\Lambda\backslash\Gamma=\{6\}$, $\Gamma\backslash\Lambda=\{2,3,7\}$ and $[1,n]\backslash(\Lambda\cup\Gamma)=\emptyset$, and then under the following designed pinning control, PBN (\ref{equation-PBN}) defined here can be stabilized to state $\delta_{512}^{44}$.

\small{\begin{equation*}\left\{\begin{array}{l}
x^1_1(\ast)=f^1_1(\ast)\vee u^1_1(\ast)\vee v_1(\ast),\\
u^1_1(\ast)=x_3\wedge \neg x_4\wedge \neg x_9, v^1_1(\ast)=\neg x_5\wedge x_9,\\
x^2_1(\ast)=f^2_1(\ast)\vee u^2_1(\ast)\vee v^2_1(\ast),\\
u^2_1(\ast)=x_3\vee x_4 \vee x_9,v^2_1(\ast)=x_4 \vee x_9,\\
x^3_1(\ast)=f^3_1(\ast)=0,\\
x_2(\ast)=f_2(\ast)\vee v_2(\ast),v_2(\ast)=\neg x_5\vee x_1\vee x_9,\\
x_3(\ast)=f_3(\ast)\vee v_3(\ast),v_3(\ast)=\neg x_2 \vee x_1,\\
x_4(\ast)=f_4(\ast)\vee u_4(\ast)\wedge v(\ast), u_4(\ast)=\neg x_6,v(\ast)=x_6, \\
x_5(\ast)=f_5(\ast)\vee u_5(\ast)\wedge v_5(\ast),\\
u_5(\ast)=\neg x_4 \wedge \neg x_9,v_5(\ast)=x_4\vee x_9,\\
x_6(\ast)=f_6(\ast)\wedge u_6(\ast),u_6(\ast)=\neg x_9,\\
x_7(\ast)=f_7(\ast)\vee v_7(\ast),v_7(\ast)=\neg x_6,\\
x_8(\ast)=f_8(\ast)\vee u_8(\ast)\wedge v_8(\ast),\\
u_8(\ast)=x_7\wedge \neg x_8\wedge(x_4\vee x_6 \vee x_9),v_8(\ast)=x_7\wedge \neg (x_4\vee x_6\vee x_9), \\
x_9(\ast)=f_9(\ast)\vee u_9(\ast)\wedge v_9(\ast), u_9(\ast)=\neg x_6\wedge x_7, v_9(\ast)=x_6.
\end{array}\right.
\end{equation*}}

\section{Conclusion}\label{section-conclusion}
This paper studied the stabilization of PBNs via pinning control based on network structure. First, the pinning nodes were selected by finding the FAS. Concerning the stochasticity of PBNs, uniform and non-uniform state feedback controllers were respectively designed to stabilize PBNs. Moreover, a criterion was obtained to determine whether each pinning node could utilize uniform controllers. Since a stable PBN may have multiple steady states, we designed another pinning control to further stabilize it at the desired state. This control design strategy was based on the network structure (local neighbors' information), rather than state transition matrix (global information). Thereby, the designed pinning control could be utilized to handle the networks with sparse connections, but not so sparse that the graph is in danger of becoming disconnected. Specially, we required $n\gg k \gg \ln(n)\gg 1$, where $k \gg \ln(n)$ guaranteed that the wiring digraph would be connected.
Eventually, the obtained results was demonstrated by a biological example about the mammalian cell-cycle encountering a mutated phenotype.



\begin{thebibliography}{10}
\bibitem{de2002modelling}
H.~de~Jong, ``Modelling and simulation of genetic regulatory systems: A literature review,'' {\em Journal of Computational Biology}, vol.~9, pp.~67--103, 2002.

\bibitem{faure2006dynamical}
A.~Faur\'{e}, A.~Naldi, C.~Chaouiya, and D.~Thieffry, ``Dynamical analysis of a generic Boolean model for the control of the mammalian cell cycle,'' {\em Bioinformatics}, vol.~22, no.~14, pp.~124--131, 2006.

\bibitem{kauffman1969metabolic}
S.~Kauffman, ``Metabolic stability and epigenesis in randomly constructed
genetic nets,'' {\em Journal of Theoretical Biology}, vol.~22, no.~3, pp.~437--467, 1969.

\bibitem{azuma1}
S. Azuma, T. Yoshida, and T. Sugie, ``Structural oscillatority analysis of Boolean networks,'' {\em IEEE Transactions on Control of Network Systems}, vol. 6, no. 2, pp. 464--473, June 2019.

\bibitem{azuma2}
S. Azuma, T. Yoshida, and T. Sugie, ``Structural monostability of activation-inhibition Boolean networks,'' {\em IEEE Transactions on Control of Network Systems}, vol. 4, no. 2, pp. 179--190, June 2015.

\bibitem{shmulevich2002probabilistic}
I.~Shmulevich, E.R.~Dougherty, S.~Kim, and W. Zhang, ``Probabilistic Boolean networks: a rule-based uncertainty model for gene regulatory networks,'' {\em Bioinformatics}, vol.~18, pp.~261--274 ,2002.

\bibitem{faryabi2008regulatory}
B.~Faryabi, G.~Vahedi, J.~F.~Chamberland, A.~Datta, and E.~R.~Dougherty, ``Optimal constrained stationary intervention in gene regulatory networks,'' {\em EURASIP Journal on Bioinformatics and Systems Biology}, vol.~1, pp.~1--10, 2008.

\bibitem{chengdz2011springer}
D.~Cheng, H.~Qi, and Z.~Li, {\em Analysis and Control of Boolean Networks: A Semi-Tensor Product Approach}. \newblock London, U.K.: Springer-Verlag, 2011.


\bibitem{cheng2010analysis}
D.~Cheng, H.~Qi, and Z.~Li, ``Analysis and control of Boolean networks: a semi-tensor product approach,'' {\em Springer Science and Business Media}, 2010.

\bibitem{zhusy2018tac}
S.~Zhu, J.~Lu, and Y.~Liu, ``Asymptotical stability of probabilistic Boolean networks with state delays,'' {\em IEEE Transactions on Automatic Control}, to be published, doi:~ 10.1109/TAC.2019.2934532.

\bibitem{liht2017siam}
H.~Li, X.~Yang, and S.~Wang, ``Perturbation Analysis for Finite-Time Stability and Stabilization of Probabilistic Boolean Networks,'' {\em IEEE Transactions on Cybernetics}, to be published, doi: 10.1109/TCYB.2020.3003055.

\bibitem{liht2019siam}
H. Li, and X. Ding, ``A control Lyapunov function approach to feedback stabilization of logical control networks," {\em SIAM Journal on Control and Optimization}, vol. 57, no. 2, pp. 810--831, 2019.

\bibitem{lu2018stabilization}
J.~Lu, L.~Sun, Y.~Liu, D.~W.~C.~Ho, and J.~Cao, ``Stabilization of Boolean
control networks under aperiodic sampled-data control,'' {\em SIAM Journal on Control and Optimization}, vol.~56, no.~6, pp.~4385--4404, 2018.

\bibitem{Margaliot2018Auto}
E. Weiss, M. Margaliot, and G. Even, ``Minimal controllability of conjunctive Boolean networks is NP-complete," {\em Automatica}, vol. 92, pp. 56--62, 2018.

\bibitem{margaliot2012aut1218}
D.~Laschov and M.~Margaliot, ``Controllability of Boolean control networks via
  Perron-Frobenius theory,'' {\em Automatica}, vol.~48, no.~6, pp.~1218--1223,
  2012.

\bibitem{yyy2019observability}
Y.~Yu, M.~Meng, and J.~Feng, ``Observability of Boolean networks via matrix equations.'' {\em Automatica}, to be publised, doi: https://doi.org/10.1016/j.automatica.2019.108621.

\bibitem{guoyq}
Y. Guo, ``Observability of Boolean control networks using
parallel extension and set reachability," {\em IEEE Transactions on
Neural Networks and Learning Systems}, vol. 29, no. 12, pp. 6402--6408,
2018.

\bibitem{guoyq2019auto}
R. Zhou, Y. Guo, and W. Gui, ``Set reachability and observability of probabilistic Boolean networks," {\em Automatica}, vol. 106, pp. 230--241, 2019.

\bibitem{chenhw2019}
H. Chen, and J. Liang, ``Local synchronization of interconnected Boolean networks with stochastic disturbances,'' {\em IEEE Transactions on Neural Networks and Learning Systems}, to be published, doi:10.1109/TNNLS.2019.2904978.

\bibitem{wu2017policy}
S.~Gao, C.~Sun, C.~Xiang, K.~Qin, and T.~Lee, ``Infinite-Horizon Optimal Control of Switched Boolean Control Networks With Average Cost: An Efficient Graph-Theoretical Approach,'' {\em IEEE Transactions on Cybernetics}, to be published, doi: 10.1109/TCYB.2020.3003552.

\bibitem{wuyhtac}
Y.~Wu, and T.~Shen, ``A finite convergence criterion for the discounted optimal control of stochastic logical networks,'' {\em IEEE Transactions on Automatic Control}, vol. 64, no. 1, pp. 262--268, 2018.

\bibitem{lirui2018siam}
R. Li, T. Chu, and X. Wang, ``Bisimulations of Boolean control networks," {\em SIAM Journal on Control and Optimization}, vol. 56, no. 1, pp. 388--416, 2018.

\bibitem{yyyTAC}
Y. Yu, J. Feng, J. Pan, and D. Cheng, ``Block decoupling of Boolean control networks," {\em IEEE Transactions on Automatic Control}, vol. 64, no. 8, pp. 3129--3140, 2019.

\bibitem{liuyang2017TAC}
Y. Liu, B. Li, J. Lu, and J. Cao, ``Pinning control for the disturbance decoupling problem of Boolean networks,'' {\em IEEE Transactions on Automatic Control}, vol. 62, no. 12, pp. 6595--6601, 2017.

\bibitem{li2013state}
R.~Li, M.~Yang, and T.~Chu, ``State feedback stabilization for Boolean control networks,'' {\em IEEE Transactions on Automatic Control}, vol.~58, no.~7, pp.~1853--1857, Jul. 2013.

\bibitem{zhusy2019}
S. Zhu, Y. Liu, Y. Lou, and J. Cao, ``Stabilization of logical control networks: An event-triggered control approach,'' {\em SCIENCE CHINA Information Sciences}, vol. 63, no. 1, pp. 163--173, 2020.

\bibitem{huang1999gene}
S.~Huang, ``Gene expression profiling, genetic networks, and cellular states: an integrating concept for tumorigenesis and drug discovery,'' {\em Journal of Molecular Medicine}, vol.~77, no.~6, pp.~469--480, 1999.

\bibitem{wu2013synchronization}
J.~Zhong, Daniel W. C. Ho, J.~Lu, and Q.~Jiao, ``Pinning Controllers for Activation Output Tracking of Boolean Network Under One-Bit Perturbation,'' {\em
IEEE Transactions on Cybernetics}, to be published, doi: 10.1109/TCYB.2018.2842819.

\bibitem{murrugarra2016identification}
D.~Murrugarra, A.~Veliz-Cuba, B.~Aguilar, and R.~Laubenbacher, ``Identification of control targets in Boolean molecular network models via computational algebra,'' {\em BMC Systems Biology}, vol.~10, no.~1, pp.~94, 2016.
\bibitem{discrete2012robert}
F.~Robert, {\em Discrete Iterations: A Metric Study}. \newblock Springer Science \& Business Media, 2012.

\bibitem{bang2008algorithems}
J.~Bang-Jensen and G.~Gutin, {\em Digraphs: Theory, Algorithms and Applications}. \newblock Springer Science \& Business Media, 2008.

\bibitem{lu2020tac}
J. Lu, J. Zhong, C. Huang, and J. Cao, ``On pinning controllability of Boolean control networks,'' {\em IEEE Transactions on Automatic Control}, voi.~61, no.~6, pp.~1658--1663, June 2016.

\bibitem{wanglq2019pinning}
Li.~Wang, M~Fang, Z. Wu, and J.~Lu, ``Necessary and Sufficient Conditions on Pinning Stabilization for Stochastic Boolean Networks,'' {\em IEEE Transactions on Cybernetics}, vol. 50, no. 10, pp. 4444--4453, October 2020.

\bibitem{liff2019set}
F.~Li and L.~Xie, ``Set stabilization of probabilistic Boolean networks using pinning control,'' {\em IEEE Transactions on Neural Networks and Learning Systems}, vol.~30, no.~8, pp.~2555--2561, August 2019.

\bibitem{liff2015pinning}
F.~Li and Y.~Tang, ``Pinning Controllability for a Boolean Network With Arbitrary Disturbance Inputs,'' {\em IEEE Transactions on Cybernetics}, to be published, doi: 10.1109/TCYB.2019.2930734.

\bibitem{zhong2020arxiv}
J. Zhong, D.W.C. Ho, and  J. Lu, ``A New Framework for Pinning Control of Boolean Networks,'' doi: arXiv:1912.01411 [eess.SY], 2019.

\bibitem{zhusy2020arxiv}
S. Zhu, J. Lu, J. Zhong, and Y. Liu, ``A novel pinning observability strategy for Boolean networks,'' doi: arXiv:1912.02394 [eess.SY], 2019.

\bibitem{RaoSolutions}
C.~Khatri and C.~Radhakrishna Rao, ``Solutions to some functional equations and their applications to characterization of probability distributions,'' {\em Sankhya: The Indian Journal of Statistics, Series A (1961-2002)}, vol.~30, no.~2, pp.~167--180, 1968.


\end{thebibliography}
\end{document}